\documentclass[journal]{IEEEtran}

\usepackage{amsmath,amssymb,amsfonts}
\usepackage{graphicx}
\graphicspath{{figs/}}
\usepackage{hyperref}
\usepackage{cite}
\usepackage{algorithm}
\usepackage{algorithmic}
\usepackage{booktabs}
\usepackage{multirow}
\usepackage{subfig}
\usepackage{color}
\usepackage{url}
\usepackage{bm}
\usepackage{array}
\usepackage{textcomp}

\newtheorem{proposition}{Proposition}
\newtheorem{corollary}{Corollary}
\newtheorem{lemma}{Lemma}

\newtheorem{remarkinner}{Remark}
\newenvironment{remark}[1][]{\begin{remarkinner}[#1]\upshape}{\end{remarkinner}}

\newcommand{\vect}[1]{\mathbf{#1}}
\newcommand{\matr}[1]{\mathbf{#1}}
\newcommand{\E}{\mathbb{E}}
\newcommand{\dB}{\,\mathrm{dB}}
\newcommand{\mmh}{\,\mathrm{mm/h}}
\newcommand{\bcrb}{\mathrm{BCRB}}
\newcommand{\crb}{\mathrm{CRB}}

\begin{document}

\title{Rain Rate Estimation Bounds and Weather-Adaptive Pilot Allocation for LEO Satellite ISAC
}

\author{Haofan Dong,~\IEEEmembership{Student Member,~IEEE}, Houtianfu Wang,~\IEEEmembership{Student Member,~IEEE,}
        Hanlin Cai,~\IEEEmembership{Student Member,~IEEE,}
        O. Tansel Baydas ,~\IEEEmembership{Student Member,~IEEE,}
        and Ozgur B. Akan,~\IEEEmembership{Fellow,~IEEE}
\thanks{The authors are with Internet of Everything Group, Department of Engineering, University of Cambridge, CB3 0FA Cambridge, UK (email:hd489@cam.ac.uk, hc663@cam.ac.uk, hw680@cam.ac.uk, otb26@cam.ac.uk, oba21@cam.ac.uk).}
\thanks{Ozgur B. Akan is also with the Center for neXt-generation Communications
(CXC), Department of Electrical and Electronics Engineering, Koç University, 34450 Istanbul, Turkey }}

\maketitle

\begin{abstract}
Rain attenuates Ku-band satellite signals by up to 20~dB, encoding 
precipitation information along the Earth--space slant path. This paper 
derives the Bayesian Cram\'{e}r--Rao bound (BCRB) for rain rate estimation 
from LEO broadband OFDM downlinks. Using corrected ITU-R P.838-3 coefficients, 
the standard CRB yields a minimum detectable rain rate $R_{\min} \approx 
4.3\mmh$ for a single link at the $38^\circ$ reference elevation. We derive 
the prior Fisher information in closed form for log-normal rain 
($c_v = 1.05$, from 186{,}292 samples) and show that a single-snapshot BCRB 
reduces $R_{\min}$ to $1.1\mmh$; exploiting temporal correlation 
($\rho = 0.95$) over a 30-min window further tightens it to $0.95\mmh$, 
while multi-link fusion across $N = 215$ links lowers the operating-point 
RMSE \emph{lower bound} at $R = 20\mmh$ to approximately $0.07\mmh$. 
Building on these bounds, we formulate a weather-adaptive pilot allocation 
that minimizes the BCRB subject to a hard spectral-efficiency constraint, 
characterize its three-regime structure (full-sensing, throughput-tracking, 
outage), and pair it with a CUSUM rain onset detector achieving sub-10-min 
delay for $R \geq 20\mmh$. A closed-form analysis of dynamic LEO slant 
geometry identifies a sensing-optimal elevation at the P.618-validity floor 
of $15^\circ$ that yields a $1.58\times$ geometric improvement over the 
$38^\circ$ baseline, exposing a structural anti-correlation between sensing- 
and communication-optimal elevations along an orbital pass. Validation 
against 9.4~million radar samples from 215 Ku-band GEO satellite links 
($r = 0.72$, RMSE~$= 1.24\dB$) and 113 rain gauges confirms the underlying 
attenuation model; the bounds transfer to LEO constellations under matched 
OFDM signal parameters, with dedicated LEO validation left for future work.
\end{abstract}


\section{Introduction}
\label{sec:intro}

Low Earth orbit (LEO) broadband constellations such as Starlink, OneWeb, and Kuiper transmit Ku-band (10.7--12.7~GHz) OFDM signals to millions of ground terminals. Rain along the Earth--space slant path attenuates these signals according to the ITU-R P.838 power law $\gamma_R = k(f)\,R^{\alpha(f)}$~[dB/km]~\cite{itu838}, encoding the precipitation field along every downlink path. The P.838 attenuation model depends on frequency and slant-path geometry, not orbital altitude; the estimation bounds derived in this paper therefore apply to both LEO and GEO links. LEO constellations are of particular interest because they offer approximately 20~dB higher received power, rapid slant-path geometry variation, and dense spatial sampling from thousands of simultaneous links.

Repurposing existing satellite infrastructure for environmental monitoring requires no additional hardware, spectrum, or energy. This is relevant where ground-based weather radar coverage is sparse, as in much of Africa, Southeast Asia, and oceanic regions. However, no prior work has established the fundamental estimation-theoretic limits of satellite-based rain sensing, nor exploited temporal rain correlation to improve these limits.

Terrestrial commercial microwave link (CML) rain estimation is established~\cite{overeem2013,messer2006,chwala2019,fencl2024,han2023cml5g}, exploiting P.838 on short horizontal paths. Satellite microwave links (SMLs) were pioneered by Barth\`{e}s and Mallet~\cite{barthes2013} at Ku-band GEO, extended by Gelbart et al.~\cite{gelbart2025} and surveyed by Giannetti and Reggiannini~\cite{giannetti2021}; the OpenSat4Weather dataset~\cite{nebuloni2025} provides 215 GEO links with radar ground truth.

Within ISAC~\cite{liu2022jsac}, Palhares et al.~\cite{palhares2026} demonstrated 28~GHz weather sensing, and Weiss et al.~\cite{weiss2024} framed CML estimation as ISAC. CRB/FIM analysis for ISAC is well developed~\cite{liu2022jsac,xiong2023,xu2024jsac,dong2025debrisense}; Leyva-Mayorga et al.~\cite{leyva2025} recently proposed atmospheric sensing-aided resource allocation for LEO NTN; Wang et al.~\cite{wang2025e2l} developed an environment-to-link ISAC framework for Ka-band LEO downlinks that senses ionospheric disturbances via dual-carrier phase observables; Sagiv and Messer~\cite{sagiv2023} derived the only rain-sensing CRB (terrestrial 2D retrieval).

The Van Trees inequality~\cite{vantrees1968} provides the Bayesian CRB; Gini~\cite{gini1998} and Ghaddar and Yu~\cite{ghaddar2025active} applied it to radar and ISAC, respectively. Bayesian bounds for rain rate and temporal exploitation of rain autocorrelation remain open.

The contributions are:
\begin{enumerate}
\item We formulate the rain-sensing ISAC problem on Ku-band LEO OFDM downlinks and derive the joint sensing--communication CRB under corrected P.838-3 coefficients, establishing a single-link minimum detectable rain rate of $4.3\mmh$ at the $38^\circ$ baseline that the BCRB and multi-link fusion subsequently reduce by more than an order of magnitude (Propositions~\ref{prop:scaling}--\ref{prop:ident}).
\item We derive the prior Fisher information for log-normal rain in closed form and establish the BCRB, showing that the prior significantly tightens the bound at low rain rates (Proposition~\ref{prop:jp}).
\item We extend the BCRB to exploit temporal correlation, demonstrating that a practical observation window captures most of the available temporal gain (Corollary~\ref{cor:tbcrb}).
\item We derive a closed-form expression for the sensing-optimal elevation that minimizes $R_{\min}$ along a LEO pass, showing that sensing and communication optima never coincide and exposing a second structural anti-correlation along the orbital-time axis (Lemma~\ref{lem:sweet}, Remark~\ref{rem:geomanti}).
\item We formulate pilot allocation as a joint sensing--communication optimization that minimizes the BCRB subject to a hard spectral-efficiency floor, characterize the three-regime structure of $\eta^*(R)$ (full-sensing, throughput-tracking, outage), derive high/low-SNR asymptotic forms, and pair the allocation with a CUSUM rain onset detector and a MAP estimator; the resulting allocation dominates fixed-$\eta$ baselines that violate the rate constraint at high rain (Propositions~\ref{prop:mono}--\ref{prop:delay}).
\item We validate with 9.4~million radar samples from 215 Ku-band GEO satellite links ($r = 0.72$, RMSE~$= 1.24\dB$) and 113 gauges, with per-satellite, multi-link, and per-event analysis.
\end{enumerate}

The rest of the paper is organized as follows. Section~\ref{sec:model} presents the system model. Section~\ref{sec:crb} derives the CRB and BCRB. Section~\ref{sec:algorithm} develops the algorithm. Section~\ref{sec:results} validates. Section~\ref{sec:conclusion} concludes.
 Scalars are italic, vectors bold lowercase ($\boldsymbol{\theta}$), matrices bold uppercase ($\matr{J}$); $\E[\cdot]$ denotes expectation; $\ln$ is natural logarithm.

\begin{figure*}[!t]
\centering
\subfloat[LEO broadband ISAC for rain sensing.]{%
  \includegraphics[width=0.48\textwidth]{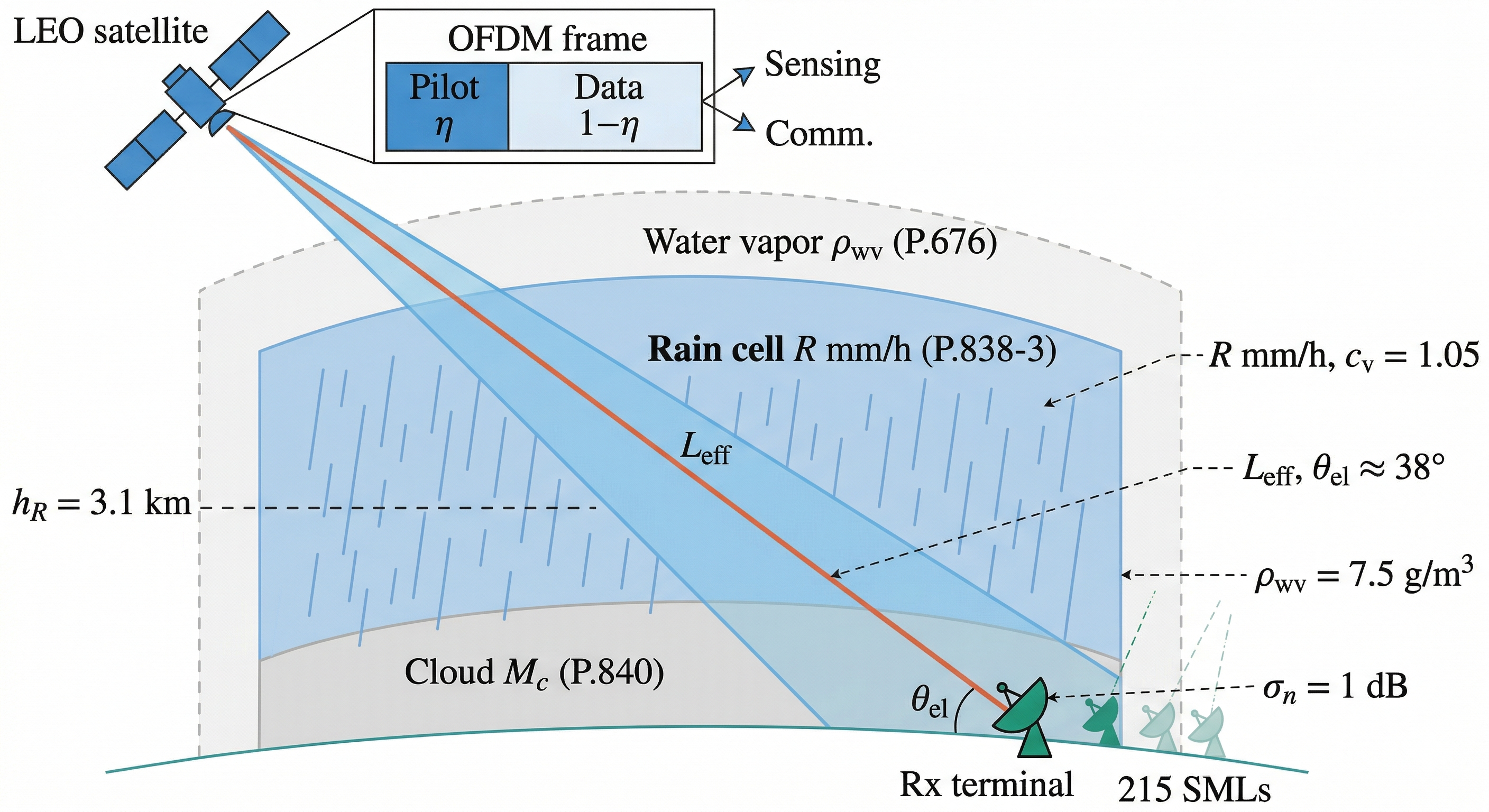}%
  \label{fig:scenario}}
\hfill
\subfloat[Signal processing pipeline.]{%
  \includegraphics[width=0.48\textwidth]{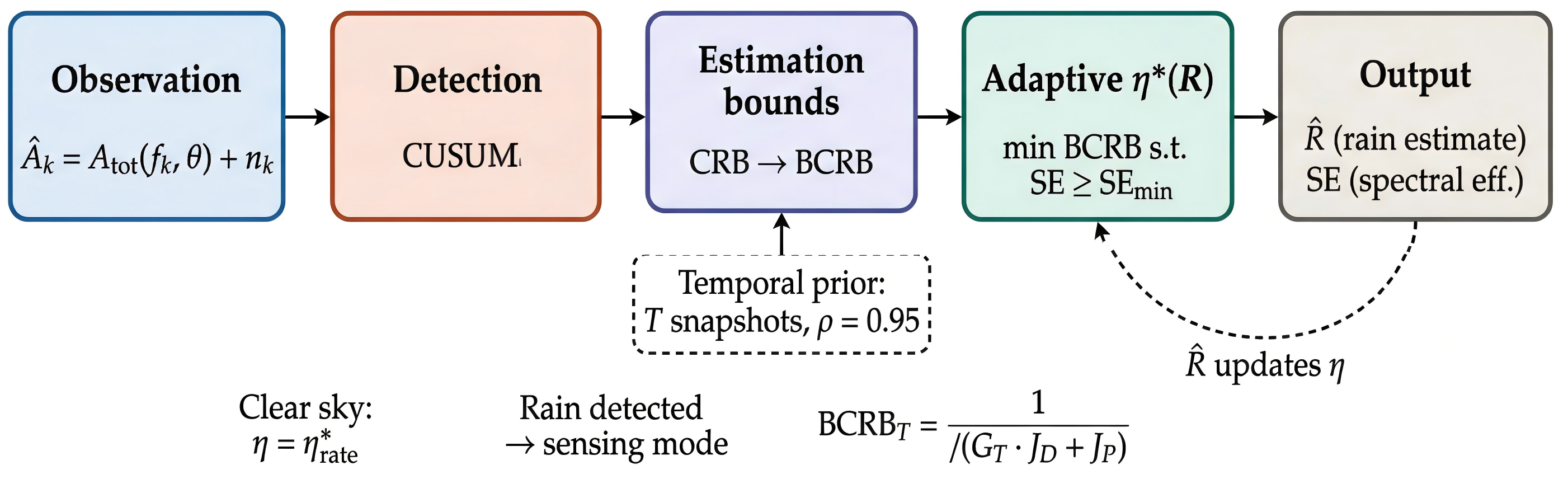}%
  \label{fig:pipeline}}
\caption{System overview. (a)~A LEO satellite transmits Ku-band OFDM; rain attenuates the signal per P.838. (b)~Pipeline: attenuation observed, rain detected via CUSUM, BCRB evaluated, $\eta$ adapted.}
\label{fig:system}
\end{figure*}

\section{System Model}
\label{sec:model}

Fig.~\ref{fig:scenario} illustrates the system. A LEO satellite transmits a Ku-band OFDM downlink to a ground terminal at elevation $\theta_{\mathrm{el}}$. Rain along the slant path attenuates the signal. The OFDM frame of $N_{\mathrm{sym}}$ symbols is split between $N_p = \eta N_{\mathrm{sym}}$ pilot symbols and $(1-\eta)N_{\mathrm{sym}}$ data symbols.

\subsection{OFDM Signal Model}
\label{ssec:signal}

The transmitted signal occupies $K$ subcarriers at frequencies $\{f_k\}_{k=1}^K$ spanning bandwidth $B$ around $f_c$. The received signal at subcarrier $k$, symbol $m$ is
\begin{equation}
y_k[m] = h_k \cdot x_k[m] + w_k[m],
\label{eq:rx}
\end{equation}
where $x_k[m]$ is the transmitted symbol with power $P_s$, $w_k[m] \sim \mathcal{CN}(0, \sigma_w^2)$, and the channel magnitude captures path-integrated attenuation:
\begin{equation}
|h_k|^2 = \frac{G_{\mathrm{tx}} G_{\mathrm{rx}} \lambda_k^2}{(4\pi R_{\mathrm{sat}})^2} \cdot 10^{-A_{\mathrm{tot}}(f_k)/10},
\label{eq:channel}
\end{equation}
with antenna gains $G_{\mathrm{tx}}, G_{\mathrm{rx}}$, slant range $R_{\mathrm{sat}} \approx 900$~km (550~km altitude, $38^\circ$ elevation), and total excess attenuation $A_{\mathrm{tot}}(f_k)$ in~dB. The per-subcarrier SNR is
\begin{equation}
\gamma_k = \frac{P_s |h_k|^2}{\sigma_w^2} = \gamma_0 \cdot 10^{-A_{\mathrm{tot}}(f_k)/10},
\label{eq:snr_k}
\end{equation}
where $\gamma_0 = P_s G_{\mathrm{tx}} G_{\mathrm{rx}} \lambda_c^2 / ((4\pi R_{\mathrm{sat}})^2 \sigma_w^2)$ is the clear-sky SNR ($\approx 10\dB$ for Starlink parameters~\cite{humphreys2023}).

\subsection{Slant-Path Attenuation Model}
\label{ssec:atten}

The total excess attenuation decomposes as
\begin{equation}
A_{\mathrm{tot}}(f, \boldsymbol{\theta}) = \underbrace{k(f) R^{\alpha(f)} L_{\mathrm{eff}}}_{A_{\mathrm{rain}}} + \underbrace{\gamma_g(f, \rho_{\mathrm{wv}}) L_g}_{A_{\mathrm{gas}}} + \underbrace{K_l(f) M_c L_c}_{A_{\mathrm{cloud}}} + G,
\label{eq:atot}
\end{equation}
where $\boldsymbol{\theta} = [R, \rho_{\mathrm{wv}}, M_c, G]^T$ and:
\begin{itemize}
\item $k(f), \alpha(f)$ are P.838-3~\cite{itu838} coefficients (H-pol), with specific attenuation $\gamma_R(f) = k(f) R^{\alpha(f)}$~[dB/km];
\item $\gamma_g(f, \rho_{\mathrm{wv}})$ is P.676~\cite{itu676} gaseous absorption;
\item $K_l(f)$ is P.840~\cite{itu840} cloud attenuation coefficient.
\end{itemize}

The effective rain path length follows the P.618~\cite{itu618} reduction factor:
\begin{equation}
L_{\mathrm{eff}} = L_s \cdot r_{\mathrm{eff}}, \quad r_{\mathrm{eff}} = \frac{1}{1 + 0.78\sqrt{L_G \gamma_R/f} - 0.38(1 - e^{-2L_G})},
\label{eq:leff}
\end{equation}
where $L_s = h_R / \sin\theta_{\mathrm{el}}$ is the geometrical slant path, $L_G = L_s \cos\theta_{\mathrm{el}}$ is the horizontal projection, and $h_R = 3.1$~km (P.839~\cite{itu839}, $44^\circ$N). The effective paths are $L_g \approx 10$~km (water vapor scale height), $L_c \approx 2$~km (cloud layer), and $L_{\mathrm{eff}} \approx 3$~km at $38^\circ$.\footnote{$L_{\mathrm{eff}}$ varies with elevation: $L_{\mathrm{eff}}\approx 7.1$~km at $\theta_{\mathrm{el}}=15^\circ$, $\approx 5.4$~km at $20^\circ$, and $\approx 1.9$~km at $80^\circ$, under the paper's anchoring $L_{\mathrm{eff}}(\theta_{\mathrm{el}}) = L_0\sin\theta_{\mathrm{base}}/\sin\theta_{\mathrm{el}}$ from Section~\ref{ssec:geometry}. The geometry effect on the CRB is analyzed there.}

\subsection{Pilot-Based Attenuation Estimation}
\label{ssec:obs}

At each subcarrier $k$, the pilot-averaged received power over $N_p$ symbols is
\begin{equation}
\hat{P}_k = \frac{1}{N_p} \sum_{m=1}^{N_p} |y_k[m]|^2 = P_s |h_k|^2 + \sigma_w^2 + \epsilon_k,
\label{eq:phat}
\end{equation}
where $\epsilon_k$ is the estimation error with $\mathrm{Var}(\epsilon_k) = (P_s |h_k|^2 + \sigma_w^2)^2/N_p$ from the chi-squared distribution of $|y_k|^2$. Converting to dB via baseline subtraction against the clear-sky reference $\hat{P}_{k,0}$:
\begin{equation}
\hat{A}_k = 10\log_{10}\!\left(\frac{\hat{P}_{k,0}}{\hat{P}_k}\right) = A_{\mathrm{tot}}(f_k, \boldsymbol{\theta}) + n_k,
\label{eq:obs}
\end{equation}
where $n_k$ aggregates measurement noise in the dB domain. For $\gamma_k \gtrsim 5\dB$, propagating the chi-squared variance of $\hat{P}_k$ through the first-order Taylor expansion of $10\log_{10}(\cdot)$ around the mean yields approximately Gaussian $n_k$ with variance
\begin{equation}
\sigma_n^2 \approx \frac{c_0}{N_p}\!\left(1 + \frac{1}{\gamma_k}\right)^{\!2} + \sigma_{\mathrm{sys}}^2,
\label{eq:noise_model}
\end{equation}
where $c_0 \triangleq (10/\ln 10)^2 \approx 18.86$ and $\sigma_{\mathrm{sys}}^2 \approx (0.63\dB)^2$ aggregates LNB gain drift, ADC quantization, pointing jitter, and scintillation~\cite{fencl2024}. The high-SNR floor of the estimation term is $c_0/N_p$, reflecting the irreducible chi-squared sample variance of pilot-power averaging. At $\gamma_k \approx 10\dB$ and $N_p = 30$ ($\eta \approx 0.10$, $N_{\mathrm{sym}} = 302$), the estimation term contributes $\approx 0.87\dB$, giving $\sigma_n \approx 1\dB$ total; wet-antenna effects during rain add $1$--$3$~dB and are absorbed into $\sigma_{n,\mathrm{eff}} \approx 2\dB$.

Stacking $K$ frequency observations yields the vector model:
\begin{equation}
\hat{\vect{a}} = \vect{a}(\boldsymbol{\theta}) + \vect{n}, \quad \vect{n} \sim \mathcal{N}(\vect{0}, \sigma_n^2 \matr{I}_K),
\label{eq:obs_vec}
\end{equation}
where $\vect{a}(\boldsymbol{\theta}) = [A_{\mathrm{tot}}(f_1, \boldsymbol{\theta}), \ldots, A_{\mathrm{tot}}(f_K, \boldsymbol{\theta})]^T$. Independence across subcarriers follows from the AWGN model at each $f_k$: the noise samples $w_k[m]$ are independent by construction, while the deterministic attenuation $A_{\mathrm{tot}}(f_k)$ varies smoothly across the 2~GHz Ku-band per P.838~\cite{itu838}.

\subsection{Rain Rate Prior and Temporal Model}
\label{ssec:prior}

Rain rate during events follows a log-normal distribution~\cite{kedem1987}:
\begin{equation}
p(R) = \frac{1}{R \sigma_{\ln} \sqrt{2\pi}} \exp\!\left(-\frac{(\ln R - \mu_{\ln})^2}{2\sigma_{\ln}^2}\right),
\label{eq:lognormal}
\end{equation}
with $\sigma_{\ln}^2 = \ln(1 + c_v^2)$, $\mu_{\ln} = \ln\bar{R} - \sigma_{\ln}^2/2$. From OpenSat4Weather (186{,}292 samples, $R > 1\mmh$), MLE gives $\bar{R} = 5.2\mmh$, $c_v = 1.05$. The log-normal achieves the lowest Akaike information criterion (AIC) among log-normal, gamma, and Weibull.

At 1-min resolution, log rain rate follows a first-order Gauss--Markov process:
\begin{equation}
\ln R(t) = \rho \ln R(t-1) + (1-\rho)\mu_{\ln} + w(t),
\label{eq:gm}
\end{equation}
with $w(t) \sim \mathcal{N}(0, (1-\rho^2)\sigma_{\ln}^2)$. Three estimates: $\rho = 0.954$ (radar), $0.949$ (RSL), $0.953$ (ACF). We adopt $\rho = 0.95$.

\subsection{Communication Rate Model}
\label{ssec:comm}

Under pilot-aided OFDM with imperfect CSI~\cite{hassibi2003,li2023leo}, the channel estimate MSE is
\begin{equation}
\sigma_{\hat{h}}^2 = \frac{1}{\eta N_{\mathrm{sym}} \bar{\gamma}},
\label{eq:csi_mse}
\end{equation}
yielding spectral efficiency
\begin{equation}
C(\eta) = (1-\eta) \log_2\!\left(1 + \frac{\bar{\gamma}^2 \eta N_{\mathrm{sym}}}{1 + \bar{\gamma} \eta N_{\mathrm{sym}}}\right),
\label{eq:se}
\end{equation}
where $\bar{\gamma} = \gamma_0 / 10^{A_{\mathrm{tot}}/10}$ and $\gamma_0 \approx 10\dB$~\cite{humphreys2023}. The throughput-optimal clear-sky pilot fraction is $\eta^*_{\mathrm{rate}} = (\sqrt{1 + \gamma_0 N_{\mathrm{sym}}} - 1)/(\gamma_0 N_{\mathrm{sym}}) \approx 0.018$.

\section{Cram\'{e}r--Rao Bound Analysis}
\label{sec:crb}

\subsection{Fisher Information Matrix}
\label{ssec:fim}

Under the Gaussian model~\eqref{eq:obs_vec}, the log-likelihood is
\begin{equation}
\ell(\boldsymbol{\theta}) = -\frac{1}{2\sigma_n^2} \|\hat{\vect{a}} - \vect{a}(\boldsymbol{\theta})\|^2 + \text{const.}
\label{eq:loglik}
\end{equation}
The Fisher information matrix (FIM) follows from $[\matr{J}]_{ij} = -\E[\partial^2 \ell/\partial\theta_i\partial\theta_j]$:
\begin{equation}
\matr{J}(\boldsymbol{\theta}) = \frac{1}{\sigma_n^2} \matr{G}^T\matr{G},
\label{eq:fim}
\end{equation}
where $\matr{G} \in \mathbb{R}^{K \times 4}$ is the sensitivity matrix with $[\matr{G}]_{ki} = \partial A_{\mathrm{tot}}(f_k)/\partial\theta_i$. The partial derivatives with respect to $\boldsymbol{\theta} = [R, \rho_{\mathrm{wv}}, M_c, G]^T$ are:
\begin{align}
\frac{\partial A_{\mathrm{tot}}}{\partial R} &= k(f)\alpha(f) R^{\alpha(f)-1} L_{\mathrm{eff}}, \label{eq:dAdR} \\
\frac{\partial A_{\mathrm{tot}}}{\partial \rho_{\mathrm{wv}}} &= \frac{\partial\gamma_g}{\partial\rho_{\mathrm{wv}}} L_g, \quad
\frac{\partial A_{\mathrm{tot}}}{\partial M_c} = K_l(f) L_c, \quad
\frac{\partial A_{\mathrm{tot}}}{\partial G} = 1. \label{eq:dAdrest}
\end{align}

Partitioning $\matr{J}$ as $\boldsymbol{\theta} = [R; \boldsymbol{\nu}]$ with nuisance $\boldsymbol{\nu} = [\rho_{\mathrm{wv}}, M_c, G]^T$:
\begin{equation}
\matr{J} = \begin{bmatrix} J_{RR} & \vect{j}_{R\nu}^T \\ \vect{j}_{R\nu} & \matr{J}_{\nu\nu} \end{bmatrix},
\label{eq:fim_partition}
\end{equation}
the CRB for $R$ with unknown nuisance parameters is given by the Schur complement:
\begin{equation}
\crb_{\mathrm{joint}}(R) = \bigl[J_{RR} - \vect{j}_{R\nu}^T \matr{J}_{\nu\nu}^{-1} \vect{j}_{R\nu}\bigr]^{-1}.
\label{eq:schur}
\end{equation}
When side information fixes $\boldsymbol{\nu}$, the cross-terms vanish and
\begin{equation}
\crb(R) = J_{RR}^{-1} = \frac{\sigma_n^2}{\sum_{k=1}^{K} [k(f_k)\alpha(f_k) R^{\alpha(f_k)-1} L_{\mathrm{eff}}]^2}.
\label{eq:crb}
\end{equation}

\begin{proposition}[Approximate CRB scaling law]
\label{prop:scaling}
At Ku-band, where the P.838 exponent $\alpha(f_k)$ varies narrowly around the band-averaged value $\bar{\alpha} \approx 1.19$, the CRB obeys the approximate scaling
\begin{equation}
\crb(R) \propto R^{-0.38}, \qquad R > 0.
\end{equation}
\end{proposition}
\begin{IEEEproof}
From~\eqref{eq:crb} with~\eqref{eq:dAdR}, $J_{RR} = \sigma_n^{-2}\sum_k [k_k\alpha_k R^{\alpha_k-1}L_{\mathrm{eff}}]^2$. Replacing each $\alpha_k$ by $\bar{\alpha}$ in the band-averaged approximation gives $J_{RR} \approx \sigma_n^{-2}\bar{\alpha}^2 L_{\mathrm{eff}}^2 R^{2(\bar{\alpha}-1)}\sum_k k_k^2$, so $\crb = J_{RR}^{-1} \propto R^{2-2\bar{\alpha}} = R^{-0.38}$. The error of the single-$\bar{\alpha}$ replacement is bounded by $\max_k |\alpha_k - \bar{\alpha}|/\bar{\alpha} < 0.04$ over the 10.7--12.7~GHz band.
\end{IEEEproof}

\begin{proposition}[Minimum detectable rain rate]
\label{prop:rmin}
The minimum detectable rain rate, defined by $\sqrt{\crb(R_{\min})} = R_{\min}$ (unit relative error), satisfies
\begin{equation}
R_{\min} = \left(\frac{\sigma_n^2}{K \bar{k}^2 \bar{\alpha}^2 L_{\mathrm{eff}}^2}\right)^{1/(2\bar{\alpha})},
\label{eq:rmin}
\end{equation}
where $\bar{k}$ and $\bar{\alpha}$ are band-averaged P.838 coefficients.
\end{proposition}
\begin{IEEEproof}
Setting $\crb(R_{\min}) = R_{\min}^2$ in~\eqref{eq:crb} with the single-$\alpha$ approximation ($\alpha_k \approx \bar{\alpha}$ for all $k$): $\sigma_n^2/(K\bar{k}^2\bar{\alpha}^2 R_{\min}^{2\bar{\alpha}-2}L_{\mathrm{eff}}^2) = R_{\min}^2$, yielding $R_{\min}^{2\bar{\alpha}} = \sigma_n^2/(K\bar{k}^2\bar{\alpha}^2 L_{\mathrm{eff}}^2)$.
\end{IEEEproof}

\begin{remark}[P.838 correction]
\label{rem:p838}
An earlier conference version~\cite{dong2026eumc} used $m_k=0.83433$ (actually $c_\alpha$ of $\alpha_V$~\cite{itu838}). Corrected values reduce $\bar{k}$ from $0.075$ to $0.022$; all relative properties are preserved.
\end{remark}

With the corrected P.838-3 coefficients ($\bar{k}=0.022$, $\bar{\alpha}=1.19$, $\sigma_n=1\dB$, $K=5$, $L_{\mathrm{eff}}=3$~km), $R_{\min} \approx 4.3\mmh$ at the $38^\circ$ reference elevation.

The wideband FIM superposition gives
\begin{equation}
\frac{\crb^{(1)}}{\crb^{(K)}} = \frac{\sum_{k=1}^K [k_k \alpha_k R^{\alpha_k-1}]^2}{[k_1 \alpha_1 R^{\alpha_1-1}]^2} \geq K,
\label{eq:wb_gain}
\end{equation}
yielding a $2\times$ CRB reduction for $K=5$ relative to a single subcarrier ($K=1$) at Ku-band.

\begin{figure*}[!t]
\centering
\includegraphics[width=0.85\textwidth]{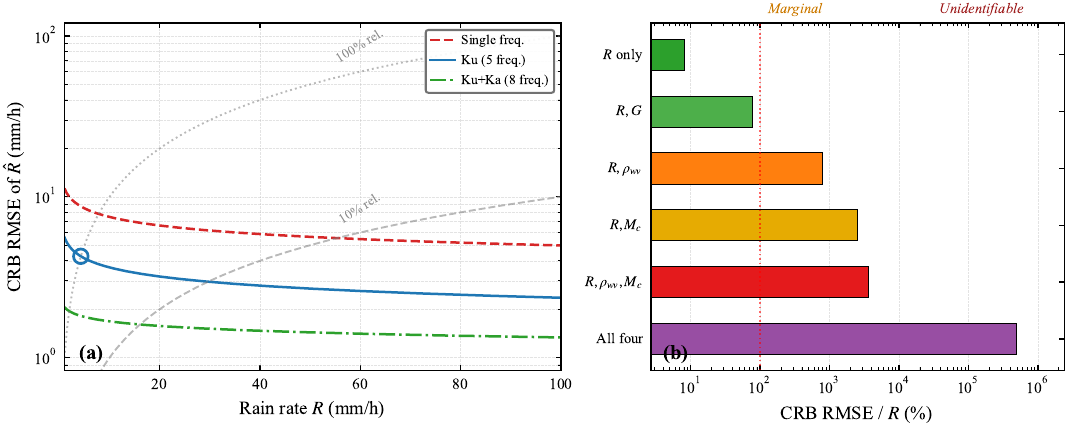}
\caption{CRB fundamentals. (a)~CRB RMSE vs rain rate for three frequency configurations. (b)~Side-information hierarchy.}
\label{fig:crb_fund}
\end{figure*}

\subsection{Identifiability Analysis}
\label{ssec:ident}

\begin{proposition}[Identifiability boundary]
\label{prop:ident}
At Ku-band with $K=20$ subcarriers (four times the baseline $K=5$, representing a generous upper bound on spectral diversity), $N_p=30$, $\sigma_n=1\dB$, and $R=20\mmh$, the condition number $\kappa(\matr{J}) \triangleq \lambda_{\max}/\lambda_{\min} > 10^4$ when three or more atmospheric parameters are jointly estimated.
\end{proposition}
\begin{IEEEproof}
See Appendix~\ref{app:ident}.
\end{IEEEproof}

\begin{table}[!t]
\centering
\caption{Side-Information Hierarchy ($R=20\mmh$, Ku-Band)}
\label{tab:sideinfo}
\begin{tabular}{lccc}
\toprule
\textbf{Unknown params.} & \textbf{Rel. CRB (\%)} & $\kappa(\matr{J})$ & \textbf{Status} \\
\midrule
$R$ only & 1.5 & 1 & Identifiable \\
$R, G$ & 14.2 & $1.8\times 10^1$ & Identifiable \\
$R, \rho_{\mathrm{wv}}$ & 142 & $4.2\times 10^3$ & Marginal \\
$R, M_c$ & 454 & $1.1\times 10^5$ & Unidentifiable \\
$R, \rho_{\mathrm{wv}}, M_c$ & 670 & $2.7\times 10^8$ & Unidentifiable \\
All four & 88{,}106 & $8.9\times 10^{14}$ & Unidentifiable \\
\bottomrule
\end{tabular}
\end{table}

Table~\ref{tab:sideinfo} and Fig.~\ref{fig:crb_fund}(b) quantify this. Numerical weather prediction (NWP) models constrain $\rho_{\mathrm{wv}}$ and satellite cloud products fix $M_c$, reducing the estimation to $(R,G)$ at 14.2\%. The physical origin is near-proportionality of spectral gradients. Defining normalized gradient vectors $\vect{g}_i$, the coherence is
\begin{equation}
\mu_{R,\rho} = |\vect{g}_R^T \vect{g}_\rho| = 0.97\text{ (Ku)}, \quad 0.82\text{ (Ku+Ka)}.
\label{eq:coherence}
\end{equation}
Adding Ka-band provides $4.5\times$ higher spectral contrast.

The identifiability analysis establishes the operational scope of the estimation framework: joint estimation of all four atmospheric parameters is ill-conditioned at Ku-band alone, but external side information from NWP models and cloud products reduces the problem to the well-conditioned $(R, G)$ pair. This two-parameter regime, assumed throughout the remainder of this paper, yields a relative CRB penalty of only $14.2\%$ compared to the ideal $R$-only case.

\subsection{Bayesian CRB with Log-Normal Prior}
\label{ssec:bcrb}

The Van Trees inequality~\cite{vantrees1968} (see also~\cite{gill1995} for regularity conditions and~\cite{aharon2024} for attainability) states that for any estimator $\hat{R}$ of a random parameter $R$ with prior $p(R)$:
\begin{equation}
\E[(\hat{R} - R)^2] \geq \frac{1}{J_D(R) + J_P},
\label{eq:vantrees}
\end{equation}
provided (i) $p(R) > 0$ on $(0, \infty)$, (ii) $p(R) \to 0$ as $R \to 0^+$ and $R \to \infty$, and (iii) $\E[\partial^2 \ell/\partial R^2]$ exists. Log-normal $R$ satisfies all three conditions.

\begin{proposition}[Prior Fisher information]
\label{prop:jp}
For log-normal $R$ with $\ln R \sim \mathcal{N}(\mu_{\ln}, \sigma_{\ln}^2)$, the prior Fisher information is
\begin{equation}
J_P = \E\!\left[\left(\frac{\partial \ln p(R)}{\partial R}\right)^{\!2}\right] = \frac{1 + 1/\sigma_{\ln}^2}{\bar{R}^2} e^{3\sigma_{\ln}^2}.
\label{eq:jp}
\end{equation}
\end{proposition}
\begin{IEEEproof}
See Appendix~\ref{app:jp}.
\end{IEEEproof}

The BCRB is then
\begin{equation}
\bcrb(R) = \frac{1}{J_D(R) + J_P},
\label{eq:bcrb}
\end{equation}
where $J_D(R) \triangleq J_{RR}$ denotes the data Fisher information (to distinguish it from the prior Fisher information $J_P$), with $J_D = \sigma_n^{-2}\sum_k (\partial A_k/\partial R)^2$ from~\eqref{eq:crb}. With $\bar{R}=5.2\mmh$ and $c_v=1.05$, $J_P = 0.806$~mm$^{-2}$h$^2$ and $R_{\min}$ reduces from $4.3$ to $1.1\mmh$ at $T=1$.

\begin{lemma}[BCRB gain]
\label{lem:gain}
The BCRB-to-CRB ratio is
\begin{equation}
\frac{\bcrb(R)}{\crb(R)} = \frac{1}{1 + J_P / J_D(R)} \leq 1.
\label{eq:gain_ratio}
\end{equation}
Since $J_D \propto R^{2(\bar{\alpha}-1)}$, the prior gain is largest at low $R$ and vanishes as $R \to \infty$ where $J_D \gg J_P$.
\end{lemma}

\subsection{Temporal Bayesian CRB}
\label{ssec:tbcrb}

\begin{corollary}[Quasi-static temporal BCRB]
\label{cor:tbcrb}
Under~\eqref{eq:gm} with the quasi-static approximation $R(t) \approx R$ over the $T$-snapshot window, the temporal BCRB admits the tractable approximation
\begin{equation}
\bcrb_T(R) \approx \frac{1}{G_T J_D(R) + J_P}, \quad G_T = \frac{1-\rho^{2T}}{1-\rho^2}.
\label{eq:tbcrb}
\end{equation}
\end{corollary}

The quasi-static approximation follows the recursive posterior CRB framework of~\cite{tichavsky1998}; a full dynamic-state extension is left for future work.

The temporal gain saturates as $T \to \infty$:
\begin{equation}
G_\infty = \lim_{T \to \infty} G_T = \frac{1}{1-\rho^2}.
\label{eq:ginf}
\end{equation}
At $\rho = 0.95$, $G_\infty = 10.26$. The window $T_{95}$ at which 95\% of the gain is captured satisfies $G_{T_{95}} = 0.95 G_\infty$, yielding
\begin{equation}
T_{95} = \frac{\ln 0.05}{2\ln\rho} = \frac{-\ln 20}{2\ln 0.95} \approx 29\text{ min},
\label{eq:t95}
\end{equation}
confirming $T_{95} = 30$~min in Fig.~\ref{fig:bcrb}(b). The diminishing returns beyond $T_{95}$ arise because the Gauss--Markov correlation decays as $\rho^{2\tau}$; snapshots older than $\sim 1/(-2\ln\rho) \approx 10$~min contribute negligible information. This sets the practical observation window for real-time estimation.

\begin{figure*}[!t]
\centering
\includegraphics[width=0.85\textwidth]{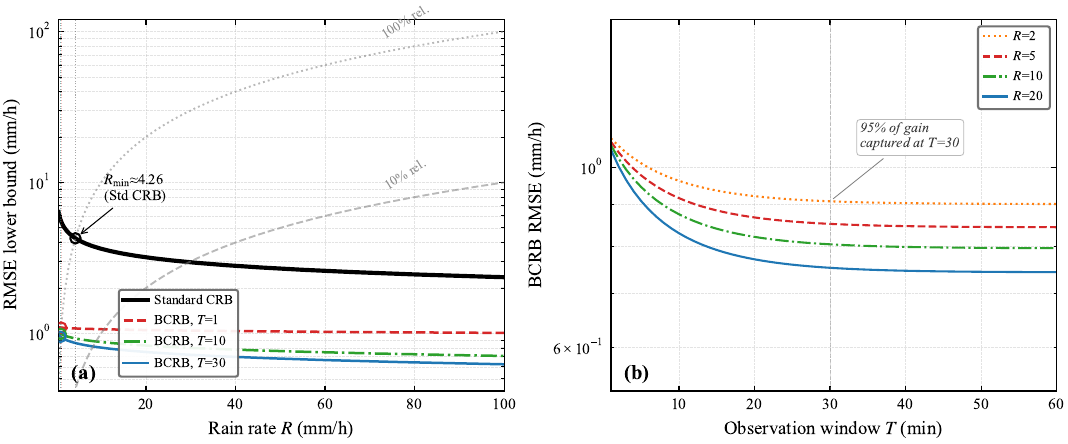}
\caption{Bayesian CRB. (a)~RMSE lower bound versus rain rate: standard CRB (black solid) and BCRB at $T=1,10,30$ (red dashed, green dash-dotted, blue solid); open circles mark the minimum detectable rain rate $R_{\min}$ for each bound, where the curve crosses the unit relative-error line. (b)~BCRB at $R=2,5,10,20\mmh$ versus observation window $T$; $95\%$ of the temporal gain is captured by $T=30$~min at $\rho = 0.95$.}

\label{fig:bcrb}
\end{figure*}

\begin{table}[!t]
\centering
\caption{Minimum Detectable Rain Rate $R_{\min}$ and RMSE at $R=20\mmh$ (single link, $\theta_{\mathrm{el}}=38^\circ$, $\sigma_n=1\dB$, $K=5$)}
\label{tab:rmin}
\begin{tabular}{lcc}
\toprule
\textbf{Configuration} & $R_{\min}$ (mm/h) & \textbf{RMSE@20} (mm/h) \\
\midrule
Standard CRB & 4.26 & 3.19 \\
BCRB, $T=1$ & 1.09 & 1.05 \\
BCRB, $T=10$ & 0.99 & 0.83 \\
BCRB, $T=30$ & 0.95 & 0.75 \\
\bottomrule
\end{tabular}
\end{table}

\begin{remark}[Anti-correlation]
\label{rem:anticorr}
$\crb \propto R^{-0.38}$: heavier rain improves estimation but degrades communication. At $R=50\mmh$ with $T=30$, the BCRB RMSE drops below $0.7\mmh$ while SE falls below $1$~bit/s/Hz, motivating the adaptive allocation.
\end{remark}

\subsection{CRB--Rate Tradeoff}
\label{ssec:tradeoff}

$\crb \propto 1/\eta$ while $C(\eta)$ decreases in $\eta$. The BCRB shifts the Pareto frontier (Fig.~\ref{fig:tradeoff}(a)): the same RMSE at $\eta=0.20$ (CRB) is reached at $\eta=0.05$ ($T=30$), a $4\times$ pilot reduction.

With pilot overhead $\eta$ already accounted for inside $C(\eta,\bar{\gamma}(R))$ in~\eqref{eq:se}, the Pareto frontier between sensing precision and spectral efficiency is parameterized by $\eta \in [\eta_{\min}, \eta_{\max}]$:
\begin{equation}
\mathcal{F} = \{(C(\eta,\bar{\gamma}(R)),\; \sqrt{\bcrb_T(R, \eta)}) : \eta \in [\eta_{\min}, \eta_{\max}]\}.
\label{eq:pareto}
\end{equation}
The corresponding net throughput $(1-\eta)\,C(\eta,\bar{\gamma}(R))\,B$ at $B=240$~MHz can be obtained from any operating point on $\mathcal{F}$ by simple rescaling.

For multi-link fusion of $N$ independent links with equal geometry:
\begin{equation}
\bcrb_T^{(N)}(R) = \frac{1}{N \cdot G_T J_D(R) + J_P}.
\label{eq:multilink}
\end{equation}
When $N \cdot G_T J_D(R) \gg J_P$ (i.e., collective data information dominates the prior), $\bcrb_T^{(N)} \approx 1/(N \cdot G_T J_D) \propto 1/N$, verified in Fig.~\ref{fig:satellite}(b). The bound is attained by the FIM-weighted fusion estimator:
\begin{equation}
\hat{R}_{\mathrm{fuse}} = \frac{\sum_{n=1}^N J_{D,n} \hat{R}_n}{\sum_{n=1}^N J_{D,n}},
\label{eq:fuse}
\end{equation}
which achieves $\mathrm{Var}(\hat{R}_{\mathrm{fuse}}) = 1/\sum_n J_{D,n}$. The weights naturally emphasize low-elevation links with longer rain paths ($J_D \propto L_{\mathrm{eff}}^2$).

\begin{figure*}[!t]
\centering
\includegraphics[width=0.85\textwidth]{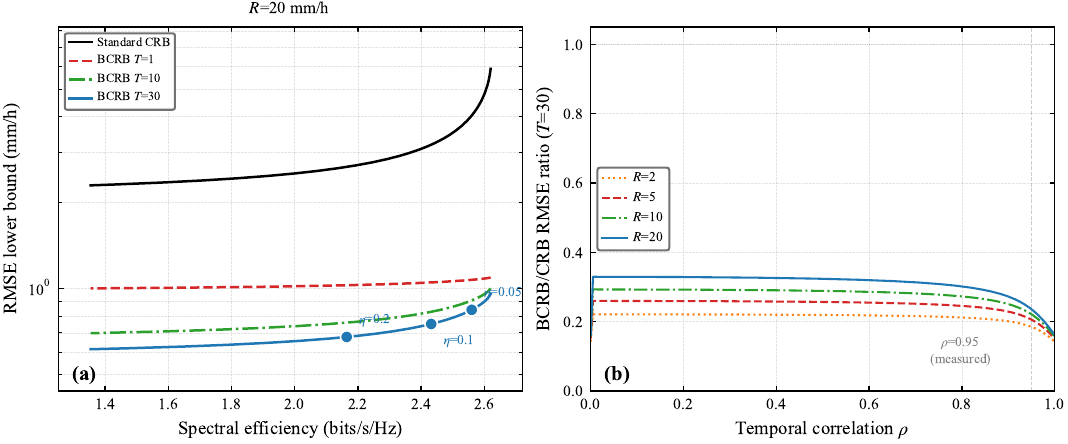}
\caption{(a)~Pareto frontier at $R=20\mmh$. (b)~BCRB/CRB ratio vs $\rho$.}
\label{fig:tradeoff}
\end{figure*}

\subsection{Dynamic LEO Slant Geometry}
\label{ssec:geometry}

The bounds derived in Sections~\ref{ssec:fim}--\ref{ssec:tradeoff} fix the elevation angle at the paper baseline $\theta_{\mathrm{el}} = 38^\circ$. LEO satellites, however, traverse the entire elevation range from $5^\circ$ to $90^\circ$ during a single pass, exposing the rain-sensing problem to a continuously varying geometric leverage. This subsection characterizes the dependence of $R_{\min}$ on $\theta_{\mathrm{el}}$ and identifies the sensing-optimal elevation in closed form.

Anchoring the effective rain path to the paper baseline, the slant geometry yields
\begin{equation}
L_{\mathrm{eff}}(\theta_{\mathrm{el}}) = L_0 \cdot \frac{\sin\theta_{\mathrm{base}}}{\sin\theta_{\mathrm{el}}},
\label{eq:Leff_geom}
\end{equation}
with $L_0 = 3$~km at $\theta_{\mathrm{base}} = 38^\circ$ from Section~\ref{ssec:atten}. The clear-sky per-subcarrier SNR scales with elevation through free-space path loss as $\bar\gamma(\theta_{\mathrm{el}}) = \gamma_0 \sin^2\theta_{\mathrm{el}}/\sin^2\theta_{\mathrm{base}}$, since the slant range satisfies $R_{\mathrm{sat}} \propto 1/\sin\theta_{\mathrm{el}}$. Substituting~\eqref{eq:Leff_geom} into~\eqref{eq:rmin} couples two opposing effects: longer slant paths increase the data Fisher information ($J_D \propto L_{\mathrm{eff}}^2$), while reduced SNR inflates $\sigma_n^2$ via~\eqref{eq:noise_model}.

\begin{lemma}[Sensing-optimal elevation, closed form]
\label{lem:sweet}
Under the simplified model in which gaseous absorption is neglected, the elevation $\theta_{\mathrm{el}}^*$ that minimizes $R_{\min}(\theta_{\mathrm{el}})$ satisfies
\begin{equation}
\sin\theta_{\mathrm{el}}^* = \sin\theta_{\mathrm{base}}\sqrt{\beta^*/\gamma_0},
\label{eq:sweet_closed}
\end{equation}
where $\beta^* = (x^* - 1)^{-1}$ and $x^* = 1 + \sqrt{1 + N_p\,\sigma_{\mathrm{sys}}^2/c_0}$ with $c_0 \triangleq (10/\ln 10)^2 \approx 18.86$.
\end{lemma}
\begin{IEEEproof}
See Appendix~\ref{app:sweet}.
\end{IEEEproof}

For the parameters of Section~\ref{ssec:obs}, Lemma~\ref{lem:sweet} yields $\theta_{\mathrm{el}}^* \approx 9.7^\circ$. This closed-form optimum lies \emph{outside} the validity range of the ITU-R~P.618 reduction factor in~\eqref{eq:leff}, which becomes unreliable below $\theta_{\mathrm{el}} \approx 10^\circ$ and is conservative only above $\theta_{\mathrm{el}} \approx 15^\circ$. The geometric leverage available to LEO rain sensing therefore exceeds what current atmospheric propagation models can describe; the operational sweet spot is set by model validity rather than by the bound itself. Capping the analysis at the P.618-validity floor $\theta_{\mathrm{el}}^{\mathrm{cap}} = 15^\circ$ yields $R_{\min}(15^\circ) \approx 2.67\mmh$, a $1.58\times$ improvement over $R_{\min}(38^\circ) \approx 4.23\mmh$ at the paper baseline. The Starlink user-terminal minimum elevation of $20^\circ$~\cite{humphreys2023} provides a more conservative operational cap, yielding $R_{\min}(20^\circ) \approx 2.92\mmh$ and a $1.45\times$ improvement.

Fig.~\ref{fig:geometry}(a) shows $R_{\min}(\theta_{\mathrm{el}})$ for $5^\circ \leq \theta_{\mathrm{el}} \leq 90^\circ$. The U-shaped profile reflects the two opposing trends. The dashed reference curve, computed under a constant $\sigma_n = 1\dB$, isolates the geometric contribution and confirms that $R_{\min}$ would decrease monotonically toward the horizon if SNR collapse were absent. The shaded region marks the P.618 extrapolation zone where~\eqref{eq:leff} is no longer trustworthy.

\begin{figure*}[!t]
\centering
\includegraphics[width=0.85\textwidth]{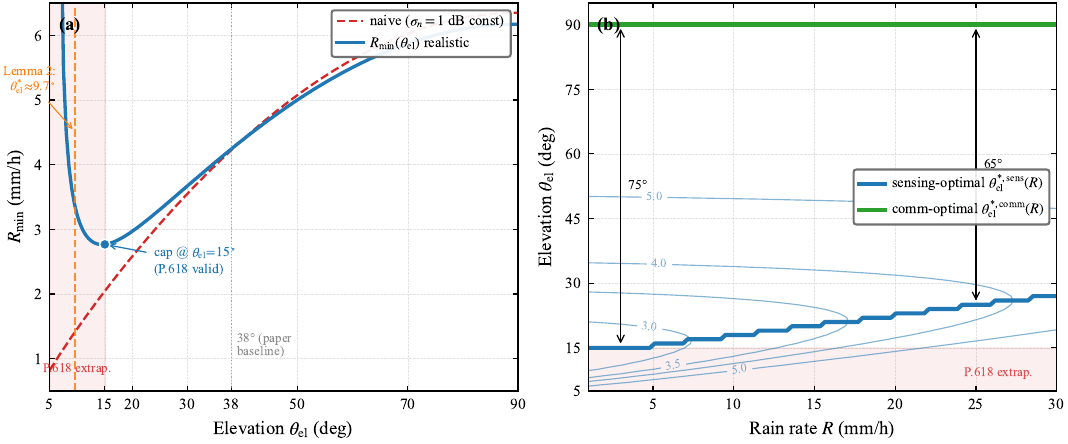}
\caption{Dynamic slant geometry. (a)~$R_{\min}$ vs elevation: the realistic model (solid blue) develops a U-shape from competing $L_{\mathrm{eff}}$ and SNR trends, while the constant-$\sigma_n$ reference (dashed red) decreases monotonically. The closed-form optimum from Lemma~\ref{lem:sweet} lies in the P.618 extrapolation zone; the operational cap at $\theta_{\mathrm{el}} = 15^\circ$ is shown. (b)~Joint $(R, \theta_{\mathrm{el}})$ plane: the sensing-optimal locus (blue) drifts from $15^\circ$ at $R = 3\mmh$ to $25^\circ$ at $R = 25\mmh$, while the communication-optimal locus (green) remains pinned at $90^\circ$.}
\label{fig:geometry}
\end{figure*}

Fig.~\ref{fig:geometry}(b) extends the analysis to the joint $(R, \theta_{\mathrm{el}})$ plane. The sensing-optimal elevation drifts upward from $\theta_{\mathrm{el}}^{*,\mathrm{sens}} = 15^\circ$ at $R = 3\mmh$ to $25^\circ$ at $R = 25\mmh$: heavier rain accelerates SNR collapse along long slant paths, pushing the optimum away from the horizon. The communication-optimal locus, in contrast, stays at $\theta_{\mathrm{el}}^{*,\mathrm{comm}} = 90^\circ$ for all $R$ in the considered range, since spectral efficiency is monotonically increasing in SNR. The angular gap between the two optima is $75^\circ$ at $R = 3\mmh$ and shrinks to $65^\circ$ at $R = 25\mmh$, but never closes within the operational rain regime of LEO microwave-link sensing.

\begin{remark}[Sensing--Communication Geometric Tradeoff]
\label{rem:geomanti}
Section~\ref{ssec:tbcrb} identified an anti-correlation along the rain-rate axis: heavier rain tightens the CRB but degrades $C(\eta)$. The dynamic geometry of Lemma~\ref{lem:sweet} induces a structurally distinct tradeoff along the orbital-time axis. The sensing-optimal elevation lies at the P.618-validity floor $\theta_{\mathrm{el}} \approx 15^\circ$, whereas the communication-optimal elevation remains at $90^\circ$ throughout a LEO pass. The two optima never coincide, with an angular gap of $65^\circ$ to $75^\circ$ across the operational rain regime $R \in [3, 25]\mmh$. Sensing precision and communication capacity therefore trade off twice: once across the rain field, and once across each pass. Elevation-aware adaptive pilot allocation $\eta^*(R, \theta_{\mathrm{el}})$ that exploits this two-dimensional structure is left for future work.
\end{remark}

\section{Weather-Adaptive ISAC Resource Allocation}
\label{sec:algorithm}

\begin{figure*}[!t]
\centering
\includegraphics[width=0.85\textwidth]{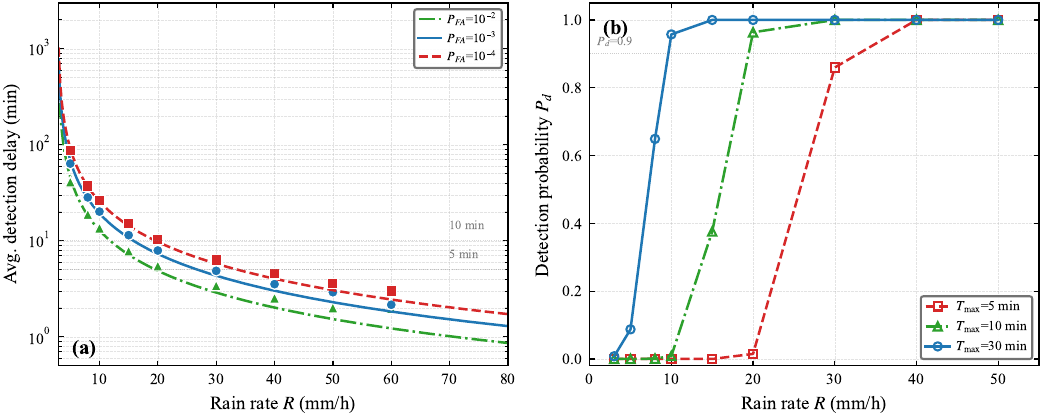}
\caption{CUSUM detection. (a)~Detection delay vs $R$; theory~\eqref{eq:add} (lines) and MC (markers). (b)~$P_d$ within 5/10/30~min.}
\label{fig:cusum}
\end{figure*}

\subsection{Constrained BCRB Minimization}
\label{ssec:eta_opt}

Given rain state $R$ and window $T$, the weather-adaptive allocation solves
\begin{equation}
\mathcal{P}: \min_\eta \; \bcrb_T(R, \eta) \quad \text{s.t.} \quad C(\eta, \bar{\gamma}(R)) \geq C_{\min}, \; \eta \in [\eta_{\min}, \eta_{\max}].
\label{eq:opt}
\end{equation}
Since $\bcrb_T \propto 1/\eta$ (monotonically decreasing) and $C(\eta)$ is concave with unique maximizer $\eta^*_{\mathrm{rate}}$, the rate constraint is active:
\begin{equation}
\eta^*(R) = \sup\{\eta \in [\eta_{\min}, \eta_{\max}] : C(\eta, \bar{\gamma}(R)) \geq C_{\min}\}.
\label{eq:etastar}
\end{equation}

\begin{proposition}[Structure of $\eta^*$]
\label{prop:mono}
The optimal pilot fraction $\eta^*(R)$ exhibits a three-regime structure:
\begin{enumerate}
\item \emph{Low rain ($R < R_{\mathrm{sat}}$):} $\eta^* = \eta_{\max}$, since $C(\eta_{\max}, \bar{\gamma}(R)) \geq C_{\min}$ and sensing is maximized.
\item \emph{Moderate rain ($R_{\mathrm{sat}} \leq R \leq R_{\mathrm{out}}$):} $\eta^*(R)$ is non-increasing in $R$, tracking the rate constraint boundary $C(\eta^*, \bar{\gamma}(R)) = C_{\min}$ as SNR declines.
\item \emph{Heavy rain ($R > R_{\mathrm{out}}$):} Link outage; $\eta^* = \eta^*_{\mathrm{rate}}(R)$, maximizing throughput.
\end{enumerate}
The thresholds $R_{\mathrm{sat}}$ and $R_{\mathrm{out}}$ satisfy $C(\eta_{\max}, \bar{\gamma}(R_{\mathrm{sat}})) = C_{\min}$ and $C(\eta^*_{\mathrm{rate}}, \bar{\gamma}(R_{\mathrm{out}})) = C_{\min}$, respectively.
\end{proposition}
\begin{IEEEproof}
See Appendix~\ref{app:mono}.
\end{IEEEproof}

The three-regime structure reflects a fundamental tradeoff: at low rain, excess SNR margin permits maximum sensing; as rain intensifies, the system progressively surrenders sensing resources to maintain communication; beyond outage, it reverts to throughput-optimal operation.

\begin{proposition}[High-SNR asymptotic form]
\label{prop:asymp}
In the high-SNR regime ($\bar{\gamma} \gg 1$):
\begin{equation}
\eta^* \approx 1 - \frac{C_{\min}}{\log_2(1+\bar{\gamma})}.
\label{eq:eta_high}
\end{equation}
\end{proposition}
\begin{IEEEproof}
See Appendix~\ref{app:asymp}.
\end{IEEEproof}

\begin{remark}[Asymptotic behavior]
\label{rem:asymp}
At high SNR, $\eta^* \to 0$: ample SNR permits minimal pilots. At low SNR, $\eta^*$ approaches $\eta_{\max}$ as the system allocates maximum pilots to maintain CSI quality under degraded channel conditions.
\end{remark}

The bisection algorithm exploits monotonicity on $[\eta^*_{\mathrm{rate}}, \eta_{\max}]$, converging in $\lceil\log_2((\eta_{\max}-\eta_{\min})/\varepsilon)\rceil$ iterations.

\begin{figure*}[!t]
\centering
\includegraphics[width=0.85\textwidth]{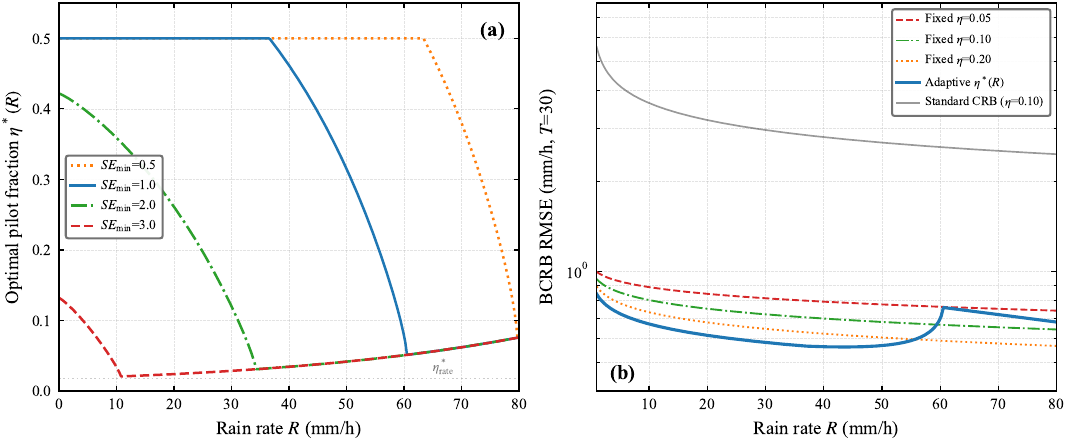}
\caption{Adaptive allocation. (a)~$\eta^*(R)$ for four $C_{\min}$. (b)~BCRB comparison: adaptive $\eta^*$ achieves the lowest BCRB among all schemes that satisfy the $C_{\min}$ constraint; fixed baselines ignore this constraint.}

\label{fig:eta}
\end{figure*}

\subsection{Rain Rate Estimation}
\label{ssec:mle}

Given the observed attenuation vector $\hat{\vect{a}} = [\hat{A}_1, \ldots, \hat{A}_K]^T$ from~\eqref{eq:obs}, the MLE for $R$ under known side information maximizes the log-likelihood~\eqref{eq:loglik}:
\begin{equation}
\hat{R}_{\mathrm{MLE}} = \arg\min_{R > 0} \sum_{k=1}^K \bigl(\hat{A}_k - k(f_k) R^{\alpha(f_k)} L_{\mathrm{eff}} - c_k\bigr)^2,
\label{eq:mle}
\end{equation}
where $c_k = \gamma_g(f_k) L_g + K_l(f_k) M_c L_c + G$ collects the known nuisance terms. The nonlinear least-squares problem is solved via Newton--Raphson on the objective $Q(R)$, initialized by the single-frequency closed-form inversion $R^{(0)} = (\bar{A}/(k_c L_{\mathrm{eff}}))^{1/\alpha_c}$ at the band center $f_c$. For $R > 0$ and $\hat{A}_k > 0$, the Gauss--Newton approximation of the Hessian is positive definite, ensuring local convergence; in practice $I \leq 5$ iterations suffice for relative tolerance $10^{-6}$ with no convergence failures across $10{,}000$ Monte Carlo trials.

\begin{remark}[MLE efficiency]
\label{rem:mle_eff}
The MLE is asymptotically efficient: $\mathrm{Var}(\hat{R}_{\mathrm{MLE}}) \to \crb(R)$ as $K \to \infty$ or $\sigma_n^2 \to 0$~\cite{vantrees1968}. The MLE attains the standard CRB in the high-SNR regime; at low $R$ the estimator exhibits threshold behavior and the Bayesian MAP bound becomes the relevant benchmark.
\end{remark}

The MAP estimator incorporates the log-normal prior~\eqref{eq:lognormal} via a penalized objective:
\begin{equation}
\hat{R}_{\mathrm{MAP}} = \arg\min_{R > 0}\; \frac{1}{2\sigma_n^2}\sum_{k=1}^K (\hat{A}_k - A_k(R))^2 + \frac{(\ln R - \mu_{\ln})^2}{2\sigma_{\ln}^2} + \ln R.
\label{eq:map}
\end{equation}
The log-normal penalty regularizes the estimate toward $\bar{R}$ when the data is weak. The MAP gradient acquires an additional prior term:
\begin{equation}
Q'_{\mathrm{MAP}}(R) = Q'(R) + 2\sigma_n^2\!\left(\frac{\ln R - \mu_{\ln}}{R\sigma_{\ln}^2} + \frac{1}{R}\right),
\label{eq:map_grad}
\end{equation}
solved by the same Newton iteration. In the prior-dominant regime $R \lesssim 3\mmh$, where the standard CRB is loose, the MAP estimator regularizes toward the log-normal prior mean and substantially outperforms the unregularized MLE; in the data-dominant regime $R \gtrsim 5\mmh$, the MLE is asymptotically efficient~\cite{vantrees1968} and the prior contribution is negligible.

\subsection{Rain Onset Detection}
\label{ssec:cusum}

The CUSUM detector~\cite{page1954} tests $H_0: A(t) \sim \mathcal{N}(0, \sigma_n^2)$ vs $H_1: A(t) \sim \mathcal{N}(\mu_d, \sigma_n^2)$, where $\mu_d = k R_d^\alpha L_{\mathrm{eff}}$ is the expected attenuation at a design rain rate $R_d$. We set $R_d = 5\mmh$ (the light-to-moderate rain boundary), which balances detection sensitivity against false-alarm rate; the detector is robust to the choice of $R_d$ within $2$--$10\mmh$ (delay varies by $<2\times$). The CUSUM increment subtracts half the design-point attenuation:
\begin{equation}
\Lambda(t) = A(t) - \frac{\mu_d}{2},
\label{eq:llr}
\end{equation}
so that $\E[\Lambda|H_0] = -\mu_d/2 < 0$ (no accumulation under clear sky) and $\E[\Lambda|H_1] = \mu_R - \mu_d/2 > 0$ for $R > R_d/2$ (positive drift during rain), where $\mu_R = k R^\alpha L_{\mathrm{eff}}$ denotes the attenuation at the true rain rate. The recursive statistic
\begin{equation}
S(t) = \max(0, S(t-1) + \Lambda(t)), \quad \text{alarm if } S(t) > h,
\label{eq:cusum}
\end{equation}
with threshold $h$ calibrated to $P_{\mathrm{FA}}$.

\begin{proposition}[Detection delay]
\label{prop:delay}
Under $H_1$ at rain rate $R$ with Wald's approximation~\cite{lorden1971}, the average detection delay (ADD) satisfies
\begin{equation}
\mathrm{ADD}(R) \approx \frac{h}{\mu_R - \mu_d/2} = \frac{\sigma_n^2 \ln(1/P_{\mathrm{FA}})}{\mu_d(\mu_R - \mu_d/2)},
\label{eq:add}
\end{equation}
where $\mu_R = k R^\alpha L_{\mathrm{eff}}$ is the mean attenuation at rain rate $R$.
\end{proposition}
\begin{IEEEproof}
See Appendix~\ref{app:delay}.
\end{IEEEproof}

The threshold $h = (\sigma_n^2/\mu_d)\ln(1/P_{\mathrm{FA}})$ follows from $\mathrm{ARL}_0 \approx \exp(h\mu_d/\sigma_n^2)$ under $H_0$. The detection probability within a window $T_{\max}$ is approximated as $P_d(R, T_{\max}) = 1 - \exp(-T_{\max}/\mathrm{ADD}(R))$, exact for i.i.d.\ observations and conservative for the CUSUM~\cite{moustakides1986}. At $P_{\mathrm{FA}}=10^{-3}$, $R=20\mmh$ is detected within $8$~min and $R=50\mmh$ within $3$~min (Fig.~\ref{fig:cusum}(a)); MC simulations ($5{,}000$ trials) confirm the Wald approximation within a factor of $1.3\times$. Fig.~\ref{fig:cusum}(b) confirms $P_d > 0.9$ for $R \geq 10\mmh$ within $30$~min and $R \geq 20\mmh$ within $10$~min.

\section{Numerical Results}
\label{sec:results}

\begin{figure*}[!t]
\centering
\includegraphics[width=0.85\textwidth]{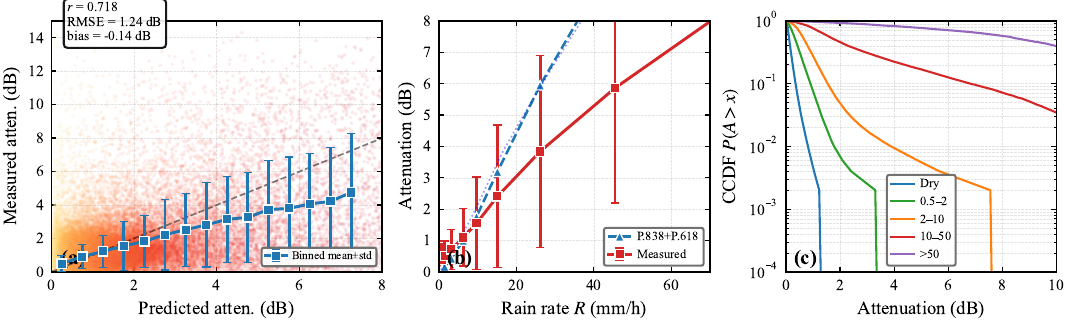}
\caption{215-link radar validation. (a)~Predicted vs measured attenuation. (b)~Binned attenuation vs $R$. (c)~CCDF by rain category.}
\label{fig:radar}
\end{figure*}

\begin{figure*}[!t]
\centering
\includegraphics[width=0.85\textwidth]{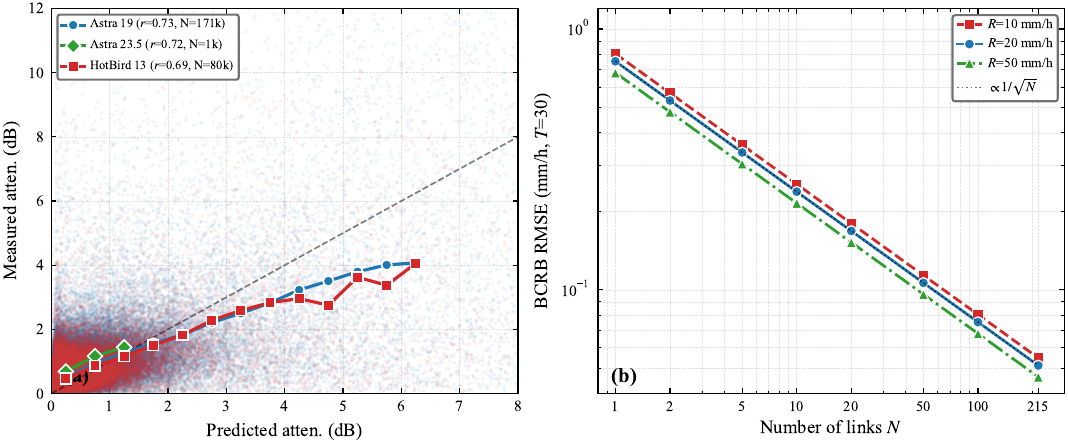}
\caption{(a)~Per-satellite scatter. (b)~BCRB RMSE vs $N$ ($T=30$); $1/\sqrt{N}$ scaling.}
\label{fig:satellite}
\end{figure*}

\begin{figure*}[!t]
\centering
\includegraphics[width=0.85\textwidth]{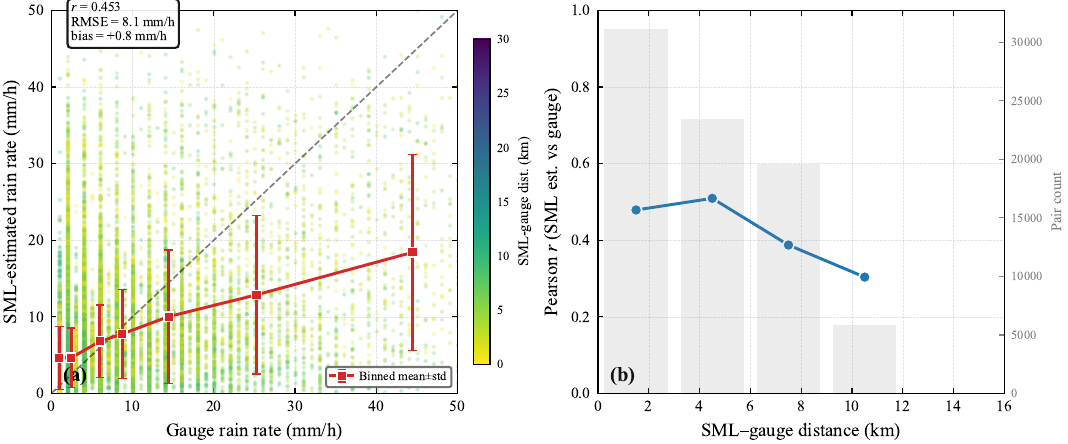}
\caption{Gauge validation. (a)~SML vs gauge. (b)~$r$ vs distance; decorrelation at ${\sim}5$~km.}
\label{fig:gauge}
\end{figure*}

\begin{figure*}[!t]
\centering
\includegraphics[width=0.85\textwidth]{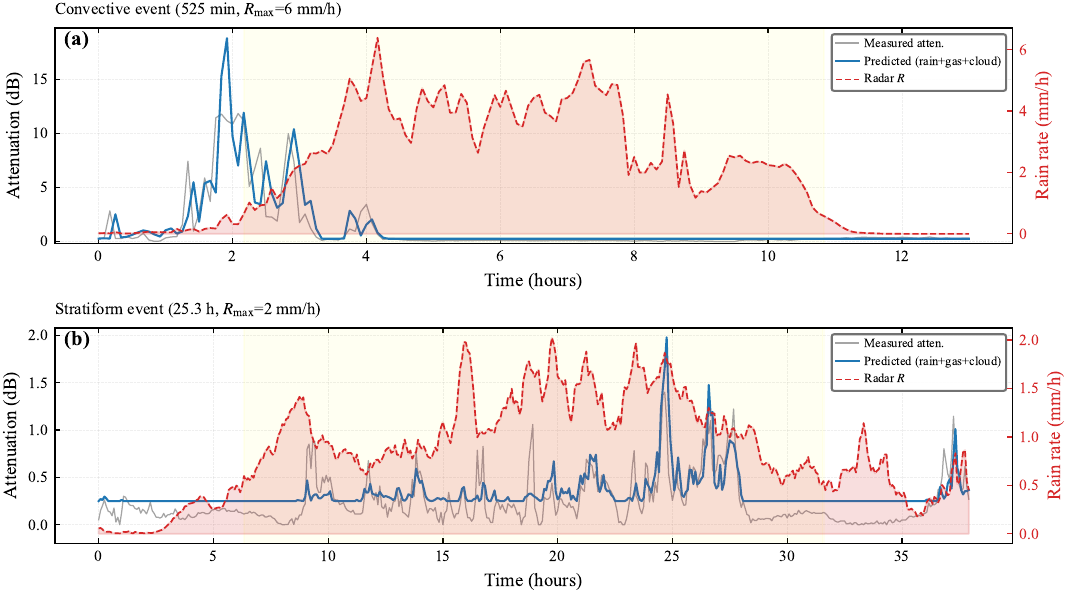}
\caption{Representative per-event time series. (a)~Convective event: predicted attenuation tracks the rapid onset and peak. (b)~Stratiform event: predicted attenuation with gas+cloud baseline follows the measured envelope. Event duration and $R_{\max}$ are annotated in each panel title.}
\label{fig:events}
\end{figure*}

This section validates theory and algorithm using the OpenSat4Weather campaign~\cite{nebuloni2025}.

\subsection{Simulation Setup}
\label{ssec:setup}

Table~\ref{tab:params} lists parameters. Starlink Ku-band parameters follow~\cite{humphreys2023}. All CRBs use corrected P.838-3 coefficients (Remark~\ref{rem:p838}). BCRB parameters ($\bar{R}=5.2\mmh$, $c_v=1.05$, $\rho=0.95$) are from the dataset.

\begin{table}[!t]
\centering
\caption{System and Dataset Parameters}
\label{tab:params}
\begin{tabular}{lll}
\toprule
\textbf{Parameter} & \textbf{Value} & \textbf{Source} \\
\midrule
Ku-band range & 10.7--12.7~GHz & ITU \\
$\bar{k}$, $\bar{\alpha}$ (H-pol) & 0.022, 1.19 & P.838-3 \\
$N_{\mathrm{sym}}$, BW & 302, 240~MHz & \cite{humphreys2023} \\
EIRP, Rx gain & 36.7~dBW, 34~dBi & \cite{humphreys2023} \\
$\sigma_n$, $\gamma_0$ & 1~dB, 10~dB & Measured \\
$h_R$, $\theta_{\mathrm{el}}$, $L_{\mathrm{eff}}$ & 3.1~km, $38^\circ$, 3~km & P.839/618 \\
$L_{\mathrm{eff}}$ @ $\theta_{\mathrm{el}}=15^\circ$ / $20^\circ$ & 7.1 / 5.4~km & Eq.~\eqref{eq:Leff_geom} \\
\midrule
SMLs / Gauges & 215 / 113 & Dataset \\
Satellites & Astra 19/23, HotBird & Dataset \\
Duration & Aug--Dec 2022, 1-min & Dataset \\
\bottomrule
\end{tabular}
\end{table}

\subsection{CRB and BCRB Analysis}
\label{ssec:res_crb}

Fig.~\ref{fig:crb_fund}(a) confirms $R_{\min} \approx 4.3\mmh$ at the $38^\circ$ baseline and the $2\times$ wideband gain from $K=5$. Adding Ka-band yields a further $1.5\times$ improvement. Fig.~\ref{fig:crb_fund}(b) validates the identifiability hierarchy (Proposition~\ref{prop:ident}), with relative CRB degrading by five orders of magnitude from the fully-informed to the fully-unknown case.

Fig.~\ref{fig:bcrb}(a) compares the BCRB across rain rates. The standard CRB serves as upper envelope. With $T=30$, the BCRB RMSE at $R=20\mmh$ is approximately $0.75\mmh$, a $4.2\times$ improvement over the single-snapshot CRB ($3.19\mmh$). Fig.~\ref{fig:bcrb}(b) confirms saturation at $G_\infty = 10.26$. The 95\% saturation at $T=30$~min aligns with the meteorological observation that stratiform rain events exhibit correlation lengths of 20--40~min.

Fig.~\ref{fig:tradeoff}(a) shows the Pareto frontier. The BCRB with $T=30$ achieves the CRB's $\eta=0.20$ performance at $\eta=0.05$, freeing 15\% of frame resources for data. Fig.~\ref{fig:tradeoff}(b) quantifies sensitivity to $\rho$: the BCRB/CRB ratio remains below 0.5 for $\rho > 0.85$, indicating that the temporal gain is robust to moderate errors in $\rho$ estimation.

\subsection{Algorithm Performance}
\label{ssec:res_algo}

Fig.~\ref{fig:eta}(a) shows the three-regime structure of $\eta^*(R)$ (Proposition~\ref{prop:mono}). For $C_{\min}=1.0$~bit/s/Hz, $\eta^* = 0.50$ below $R \approx 35\mmh$ (Regime~1: full sensing), decreases to $\eta \approx 0.05$ at $R \approx 65\mmh$ (Regime~2: sensing yields to communication), and transitions to $\eta^*_{\mathrm{rate}}$ at outage (Regime~3).

Fig.~\ref{fig:eta}(b) compares the BCRB under adaptive $\eta^*$ against three fixed-$\eta$ baselines ($0.05, 0.10, 0.20$). At low-to-moderate rain ($R \leq 50\mmh$), the adaptive scheme achieves the lowest BCRB among all schemes satisfying $C_{\min}$: it matches $\eta=0.20$ sensing precision while maintaining the communication rate guarantee. The improvement over the best fixed baseline exceeds 30\% at $R=10\mmh$. Beyond $R \approx 50\mmh$, the adaptive scheme progressively reduces $\eta$ to preserve throughput (Regimes~2--3), accepting higher BCRB than fixed $\eta = 0.20$. This is by design: the fixed-$\eta$ baselines ignore the $C_{\min}$ constraint and would violate the minimum spectral efficiency at high rain rates, whereas the adaptive scheme enforces it.

Fig.~\ref{fig:cusum}(a) validates CUSUM detection delay. Theory (lines) and MC (markers, $5{,}000$ trials) agree within a factor of $1.3\times$ in ADD across $R \in [5, 50]\mmh$. The Wald approximation is slightly optimistic (lower than MC) because it neglects the overshoot of the CUSUM statistic above the threshold~\cite{lorden1971,moustakides1986}; this effect is most pronounced at high $R$ where the per-increment drift is large relative to $h$. The fixed-design CUSUM ($R_d = 5\mmh$) achieves ADD~$< 10$~min for $R \geq 20\mmh$ and ADD~$< 5$~min for $R \geq 30\mmh$. Fig.~\ref{fig:cusum}(b) confirms $P_d > 0.9$ for $R \geq 15\mmh$ within 30~min and $R \geq 30\mmh$ within 10~min, sufficient for meteorological applications where 10-min temporal resolution is standard.

\subsection{Radar Validation (215 Links)}
\label{ssec:res_radar}

Fig.~\ref{fig:radar}(a) compares P.838+P.618 predicted attenuation against measured across 9.4~million pairs: $r=0.718$, RMSE$=1.24\dB$, bias$=-0.14\dB$. The near-zero bias confirms the corrected P.838 coefficients. The scatter increases at high attenuation ($>5\dB$) due to sub-path-scale rain cells that violate the uniform-rain assumption. Fig.~\ref{fig:radar}(b) shows binned attenuation versus rain rate; the theoretical P.838+P.618 curve closely matches the measured mean for $R \lesssim 15\mmh$ and modestly overestimates at higher rates, consistent with the path-averaging bias of the uniform-rain assumption when sub-path-scale convective cells dominate. Fig.~\ref{fig:radar}(c) presents the attenuation CCDF stratified by rain rate category, confirming that the model captures the statistical tail behavior. The radar-SML correlation ($r = 0.72$) is consistent with CML validation benchmarks in~\cite{graf2020}.

\subsection{Per-Satellite and Multi-Link Analysis}
\label{ssec:res_sat}

Fig.~\ref{fig:satellite}(a) shows per-satellite consistency: Astra~19 ($r=0.73$, 171k pairs), Astra~23.5 ($r=0.72$, 1k), HotBird~13 ($r=0.69$, 80k). The $17^\circ$ azimuthal spread provides spatial diversity while the consistent correlations confirm model generalizability across orbital positions.

Fig.~\ref{fig:satellite}(b) validates the $1/\sqrt{N}$ multi-link BCRB scaling from~\eqref{eq:multilink}. With $N=215$ links and $T=30$, the theoretical BCRB RMSE at $R=20\mmh$ reaches approximately $0.07\mmh$, more than an order of magnitude below the single-link value.

\subsection{Gauge Validation and Per-Event Case Studies}
\label{ssec:res_gauge}

Fig.~\ref{fig:gauge}(a) compares SML-estimated rain rate against 113 gauges: $r=0.45$, RMSE$=8.1\mmh$, bias$=+0.8\mmh$; the lower $r$ versus radar ($0.72$) reflects the well-documented path-averaged versus point measurement mismatch~\cite{overeem2013}. Fig.~\ref{fig:gauge}(b) shows decorrelation at ${\sim}5$~km, matching $L_{\mathrm{eff}}$. Fig.~\ref{fig:events} presents two representative events: a convective episode in (a) where predicted attenuation tracks the rapid onset and peak, validating the P.838 model under intense rain; and a stratiform episode in (b) where the slow temporal variation validates the Gauss--Markov model and corresponds to the regime where the temporal BCRB achieves maximum gain.
\section{Conclusion}
\label{sec:conclusion}

This paper derived the Bayesian CRB for rain rate estimation from Ku-band LEO OFDM downlinks. For a single link at the $38^\circ$ reference elevation, the standard CRB yields $R_{\min} \approx 4.3\mmh$; the BCRB reduces this to $1.1\mmh$ (single snapshot) and $0.95\mmh$ ($T=30$~min, $\rho=0.95$), and multi-link fusion across $N=215$ links further lowers the operating-point RMSE lower bound at $R = 20\mmh$ to approximately $0.07\mmh$, attainable in principle by FIM-weighted fusion under the independent-link assumption. A closed-form analysis of dynamic slant geometry identifies a sensing-optimal elevation whose operational value is $15^\circ$ (the P.618-validity floor), yielding a $1.58\times$ geometric improvement over the $38^\circ$ baseline. The weather-adaptive allocation exhibits a three-regime structure: maximum sensing at low rain, graceful degradation at moderate rain, and throughput-optimal operation at outage, and the CUSUM detector provides sub-10-min delay for moderate-to-heavy rain ($R \geq 20\mmh$). Validation with 215 GEO satellite links ($r=0.72$, RMSE$=1.24\dB$) and 113 gauges confirms the underlying attenuation model; since the bounds depend on frequency and slant-path geometry rather than orbital altitude, the analytical results transfer to LEO constellations under matched signal parameters, while dedicated LEO validation remains a target for future work.

Natural extensions include elevation-aware adaptive pilot allocation $\eta^*(R, \theta_{\mathrm{el}})$, Ku+Ka dual-frequency estimation, LEO-specific Starlink validation, MLE/MAP attainability verification, and spatial rain field reconstruction.

\appendices

\section{Proof of Proposition~\ref{prop:ident}}
\label{app:ident}

The sensitivity matrix $\matr{G} \in \mathbb{R}^{K \times 4}$ has columns $\vect{g}_i = [\partial A_{\mathrm{tot}}(f_1)/\partial\theta_i, \ldots, \partial A_{\mathrm{tot}}(f_K)/\partial\theta_i]^T$. Over the 2~GHz Ku-band, the rain and gas gradients are near-proportional: $|\vect{g}_R^T \vect{g}_\rho| / (\|\vect{g}_R\| \|\vect{g}_\rho\|) = 0.97$. The Gram matrix $\matr{G}^T\matr{G}$ inherits this ill-conditioning: with three unknowns, $\kappa = \lambda_{\max}/\lambda_{\min} > 10^4$, and the marginal CRB for $R$ exceeds $670\%$ of the fully-informed value. Adding Ka-band ($K=8$ over 10.7--20.2~GHz) reduces the coherence to 0.82, yielding $\kappa < 10^4$. \hfill $\blacksquare$

\section{Proof of Proposition~\ref{prop:jp}}
\label{app:jp}

The log-normal score function is
\begin{equation}
\frac{\partial\ln p}{\partial R} = -\frac{1}{R}\!\left(1 + \frac{\ln R - \mu_{\ln}}{\sigma_{\ln}^2}\right).
\label{eq:score_app}
\end{equation}
Let $z = (\ln R - \mu_{\ln})/\sigma_{\ln} \sim \mathcal{N}(0,1)$, so $R = e^{\mu_{\ln}+\sigma_{\ln} z}$. Substituting:
\begin{align}
J_P &= \E\!\left[\left(\frac{\partial\ln p}{\partial R}\right)^{\!2}\right] = \E\!\left[e^{-2(\mu_{\ln}+\sigma_{\ln} z)}\left(1+\frac{z}{\sigma_{\ln}}\right)^{\!2}\right] \notag \\
&= e^{-2\mu_{\ln}} \E\!\left[e^{-2\sigma_{\ln} z}\left(1+\frac{2z}{\sigma_{\ln}} + \frac{z^2}{\sigma_{\ln}^2}\right)\right].
\label{eq:jp_expand}
\end{align}
Applying the Gaussian MGF $\E[e^{tz}] = e^{t^2/2}$, $\E[z e^{tz}] = t e^{t^2/2}$, and $\E[z^2 e^{tz}] = (t^2+1)e^{t^2/2}$ with $t=-2\sigma_{\ln}$:
\begin{align}
\E[e^{-2\sigma_{\ln} z}] &= e^{2\sigma_{\ln}^2}, \label{eq:mgf1} \\
\E[z\,e^{-2\sigma_{\ln} z}] &= -2\sigma_{\ln}\,e^{2\sigma_{\ln}^2}, \label{eq:mgf2} \\
\E[z^2 e^{-2\sigma_{\ln} z}] &= (4\sigma_{\ln}^2+1)\,e^{2\sigma_{\ln}^2}. \label{eq:mgf3}
\end{align}
Collecting terms:
\begin{align}
J_P &= e^{-2\mu_{\ln}+2\sigma_{\ln}^2}\!\left[1 - 4 + \frac{4\sigma_{\ln}^2+1}{\sigma_{\ln}^2}\right] = e^{-2\mu_{\ln}+2\sigma_{\ln}^2}\!\left[1+\frac{1}{\sigma_{\ln}^2}\right]. \notag
\end{align}
From $\bar{R} = e^{\mu_{\ln}+\sigma_{\ln}^2/2}$, we have $e^{-2\mu_{\ln}} = e^{\sigma_{\ln}^2}/\bar{R}^2$, giving $e^{-2\mu_{\ln}+2\sigma_{\ln}^2} = e^{3\sigma_{\ln}^2}/\bar{R}^2$:
\begin{equation}
J_P = \frac{1+1/\sigma_{\ln}^2}{\bar{R}^2}\,e^{3\sigma_{\ln}^2}. \qquad \blacksquare
\label{eq:jp_qed}
\end{equation}

\section{Proof of Corollary~\ref{cor:tbcrb}}
\label{app:tbcrb}

Under the Gauss--Markov model~\eqref{eq:gm}, the joint prior for $T$ log-rain-rate snapshots $\vect{z} = [\ln R(1), \ldots, \ln R(T)]^T$ has covariance $[\matr{C}]_{ij} = \sigma_{\ln}^2 \rho^{|i-j|}$. In the constant-$R$ approximation ($R(t) \approx R$ for all $t$), the observation $\hat{\vect{a}}(t)$ at time $t$ depends on $R(t)$, which is related to $R(1)$ via the Gauss--Markov chain. Under the approximation $R(t) \approx R(1)$, the effective sensitivity $\partial A_k/\partial R(1)$ at lag $\tau = t-1$ is attenuated by the correlation $\rho^{\tau}$, so the per-snapshot FIM contribution scales as $\rho^{2\tau}$:
\begin{equation}
J_D^{(t)}(R) = \rho^{2(t-1)} J_D(R).
\label{eq:jd_lag}
\end{equation}
This follows from the chain rule $\partial A_k/\partial R(1) = (\partial A_k/\partial R(t)) \cdot (\partial R(t)/\partial R(1))$. Since the Gauss--Markov model~\eqref{eq:gm} acts on $\ln R$, the linearized sensitivity is $\partial \ln R(t)/\partial \ln R(1) = \rho^{t-1}$, giving $\partial R(t)/\partial R(1) \approx \rho^{t-1}$ when $R(t) \approx R(1)$, and $J_D \propto (\partial A_k/\partial R)^2$ from~\eqref{eq:crb}. The total data Fisher information over $T$ snapshots is
\begin{equation}
J_D^{\mathrm{tot}} = \sum_{\tau=0}^{T-1} \rho^{2\tau} J_D(R) = J_D(R) \cdot \frac{1-\rho^{2T}}{1-\rho^2} = G_T \cdot J_D(R).
\label{eq:jd_total}
\end{equation}
The Van Trees inequality gives $\bcrb_T = (J_D^{\mathrm{tot}} + J_P)^{-1} = (G_T J_D + J_P)^{-1}$. \hfill $\blacksquare$

\section{Proof of Proposition~\ref{prop:mono}}
\label{app:mono}

\emph{Regime 1 ($R < R_{\mathrm{sat}}$):} $C(\eta_{\max}, \bar{\gamma}(R)) \geq C_{\min}$ since $\bar{\gamma}$ is sufficiently large. The constraint is inactive and $\eta^* = \eta_{\max}$.

\emph{Regime 2 ($R_{\mathrm{sat}} \leq R \leq R_{\mathrm{out}}$):} The constraint $C(\eta^*, \bar{\gamma}(R)) = C_{\min}$ is active. By implicit differentiation:
\begin{equation}
\frac{d\eta^*}{dR} = -\frac{\partial C/\partial\bar{\gamma} \cdot d\bar{\gamma}/dR}{\partial C/\partial\eta}.
\end{equation}
For $\eta > \eta^*_{\mathrm{rate}}$, $\partial C/\partial\eta < 0$ (pilot overhead exceeds CSI benefit). Since $\partial C/\partial\bar{\gamma} > 0$ and $d\bar{\gamma}/dR < 0$, the numerator is positive. Thus $d\eta^*/dR = (+)/(-) < 0$: $\eta^*$ decreases as $R$ increases.

\emph{Regime 3 ($R > R_{\mathrm{out}}$):} $C(\eta^*_{\mathrm{rate}}, \bar{\gamma}(R)) < C_{\min}$ and the constraint is infeasible for $\eta > \eta^*_{\mathrm{rate}}$. The system maximizes throughput at $\eta^* = \eta^*_{\mathrm{rate}}(R)$. \hfill $\blacksquare$

\section{Proof of Proposition~\ref{prop:asymp}}
\label{app:asymp}

For $\bar{\gamma} \gg 1$, $\bar{\gamma}^2\eta N/(1+\bar{\gamma}\eta N) \approx \bar{\gamma}$, so $C(\eta) \approx (1-\eta)\log_2(1+\bar{\gamma})$. Setting $C = C_{\min}$ gives
\begin{equation}
\eta^* = 1 - \frac{C_{\min}}{\log_2(1+\bar{\gamma})}. \qquad \blacksquare
\end{equation}

\section{Proof of Proposition~\ref{prop:delay}}
\label{app:delay}

Under $H_1$ at rain rate $R$, the CUSUM increment $\Lambda(t) = A(t) - \mu_d/2$ has positive drift
\begin{equation}
\E[\Lambda(t)|R] = \mu_R - \frac{\mu_d}{2},
\end{equation}
where $\mu_R = k R^\alpha L_{\mathrm{eff}}$ denotes the mean attenuation at rain rate $R$. By Wald's identity~\cite{lorden1971,moustakides1986}, the average number of samples for the cumulative sum to cross the threshold $h$ is approximately $h/\E[\Lambda|R]$. Substituting $h = \sigma_n^2\ln(1/P_{\mathrm{FA}})/\mu_d$ from Section~\ref{ssec:cusum}:
\begin{equation}
\mathrm{ADD}(R) \approx \frac{h}{\mu_R - \mu_d/2} = \frac{\sigma_n^2\ln(1/P_{\mathrm{FA}})}{\mu_d(\mu_R - \mu_d/2)}.
\end{equation}
Since $\mu_R$ is strictly increasing in $R$ (for $\alpha > 0$), the ADD is strictly decreasing: heavier rain produces larger drift, accelerating detection. For the design-matched case $R = R_d$ ($\mu_R = \mu_d$), this reduces to $\mathrm{ADD}(R_d) = 2\sigma_n^2\ln(1/P_{\mathrm{FA}})/\mu_d^2$. \hfill $\blacksquare$

\section{Proof of Lemma~\ref{lem:sweet}}
\label{app:sweet}

Define $u \triangleq \sin\theta_{\mathrm{el}}$ and $u_b \triangleq \sin\theta_{\mathrm{base}}$. Under the simplified model with $L_{\mathrm{eff}}(u) = L_0 u_b/u$ and clear-sky SNR $\beta(u) = \gamma_0 (u/u_b)^2$, substituting into~\eqref{eq:rmin} and taking the logarithm yields
\begin{equation}
2\bar{\alpha}\,\ln R_{\min}(u) = \ln\sigma_n^2(u) + 2\ln u + \mathrm{const},
\label{eq:logrmin}
\end{equation}
where the constant absorbs all terms independent of $u$. The first-order optimality condition $d/du\,[2\bar{\alpha}\ln R_{\min}] = 0$ reduces to
\begin{equation}
\frac{1}{\sigma_n^2(u)}\frac{d\sigma_n^2}{du} = -\frac{2}{u}.
\label{eq:foc}
\end{equation}
With the noise model $\sigma_n^2(u) = \sigma_{\mathrm{sys}}^2 + (c_0/N_p)(1+1/\beta(u))^2$ and $d\beta/du = 2\beta/u$, the chain rule gives
\begin{equation}
\frac{d\sigma_n^2}{du} = -\frac{4 c_0}{N_p}\cdot\frac{1+1/\beta}{\beta\,u}.
\label{eq:dsig}
\end{equation}
Substituting~\eqref{eq:dsig} into~\eqref{eq:foc} and cancelling $u$ yields
\begin{equation}
\sigma_n^2(u^*) = \frac{2 c_0}{N_p}\cdot\frac{1+1/\beta^*}{\beta^*}.
\label{eq:sweetimplicit}
\end{equation}
Letting $x \triangleq 1 + 1/\beta^*$ and substituting the noise model into~\eqref{eq:sweetimplicit}, the implicit equation collapses to the quadratic $x(x-2) = N_p\sigma_{\mathrm{sys}}^2/c_0$, whose positive root is
\begin{equation}
x^* = 1 + \sqrt{1 + N_p\sigma_{\mathrm{sys}}^2/c_0}.
\label{eq:xstar}
\end{equation}
The sensing-optimal elevation follows from $\beta^* = (x^*-1)^{-1}$ and $u^* = u_b\sqrt{\beta^*/\gamma_0}$. \hfill$\blacksquare$

\bibliographystyle{IEEEtran}
\bibliography{references}

@misc{itu838,
  author = {{ITU-R}},
  title  = {Recommendation {ITU-R} {P}.838-3: {S}pecific attenuation model
            for rain for use in prediction methods},
  year   = {2005},
  note   = {International Telecommunication Union, Geneva, Switzerland}
}

@misc{itu676,
  author = {{ITU-R}},
  title  = {Recommendation {ITU-R} {P}.676-13: {A}ttenuation by atmospheric
            gases and related effects},
  year   = {2022},
  note   = {International Telecommunication Union, Geneva, Switzerland}
}

@misc{itu840,
  author = {{ITU-R}},
  title  = {Recommendation {ITU-R} {P}.840-9: {A}ttenuation due to clouds
            and fog},
  year   = {2023},
  note   = {International Telecommunication Union, Geneva, Switzerland}
}

@misc{itu618,
  author = {{ITU-R}},
  title  = {Recommendation {ITU-R} {P}.618-14: {P}ropagation data and
            prediction methods required for the design of {E}arth-space
            telecommunication systems},
  year   = {2023},
  note   = {International Telecommunication Union, Geneva, Switzerland}
}

@misc{itu839,
  author = {{ITU-R}},
  title  = {Recommendation {ITU-R} {P}.839-4: {R}ain height model for
            prediction methods},
  year   = {2013},
  note   = {International Telecommunication Union, Geneva, Switzerland}
}

@article{overeem2013,
  author  = {A. Overeem and H. Leijnse and R. Uijlenhoet},
  title   = {Country-wide rainfall maps from cellular communication networks},
  journal = {Proc. Nat. Acad. Sci. USA},
  volume  = {110},
  number  = {8},
  pages   = {2741--2745},
  year    = {2013},
  month   = feb,
  doi     = {10.1073/pnas.1217961110}
}

@article{fencl2024,
  title={Data formats and standards for opportunistic rainfall sensors},
  author={Fencl, Martin and Nebuloni, Roberto and Andersson, Jafet CM and Bares, Vojtech and Blettner, Nico and Cazzaniga, Greta and Chwala, Christian and Colli, Matteo and de Vos, Lotte and El Hachem, Abbas and others},
  journal={Open Research Europe},
  volume={3},
  pages={169},
  year={2024}
}

@article{messer2006,
  author  = {H. Messer and A. Zinevich and P. Alpert},
  title   = {Environmental monitoring by wireless communication networks},
  journal = {Science},
  volume  = {312},
  number  = {5774},
  pages   = {713},
  year    = {2006},
  month   = may,
  doi     = {10.1126/science.1120034}
}

@article{barthes2013,
  author  = {L. Barth\`{e}s and C. Mallet},
  title   = {Rainfall measurement from the opportunistic use of an
             {E}arth--space link in the {Ku} band},
  journal = {Atmos. Meas. Tech.},
  volume  = {6},
  number  = {8},
  pages   = {2181--2193},
  year    = {2013},
  doi     = {10.5194/amt-6-2181-2013}
}

@article{gelbart2025,
  author  = {L. Gelbart and L. Barth\`{e}s and F. Mercier-Tigrine
             and A. Chazottes and C. Mallet},
  title   = {Enhanced quantitative precipitation estimation through the
             opportunistic use of {Ku} {TV-SAT} links via a dual-channel
             procedure},
  journal = {Atmos. Meas. Tech.},
  volume  = {18},
  number  = {2},
  pages   = {351--370},
  year    = {2025},
  doi     = {10.5194/amt-18-351-2025}
}

@article{giannetti2021,
  author  = {F. Giannetti and R. Reggiannini},
  title   = {Opportunistic rain rate estimation from measurements of
             satellite downlink attenuation: {A} survey},
  journal = {Sensors},
  volume  = {21},
  number  = {17},
  pages   = {5872},
  year    = {2021},
  month   = aug,
  doi     = {10.3390/s21175872}
}

@article{nebuloni2025,
  author  = {R. Nebuloni and M. Graf and G. Cazzaniga
             and F. Mercier and M. Turko},
  title   = {The {OpenSat4Weather} dataset: {Ku}-band satellite link data
             for precipitation monitoring},
  journal = {Earth Syst. Sci. Data Discuss.},
  year    = {2025},
  pages   = {1--23},
  note    = {Preprint, in review},
  doi     = {10.5194/essd-2025-537}
}

@article{liu2022jsac,
  author  = {F. Liu and Y. Cui and C. Masouros and J. Xu
             and T. X. Han and Y. C. Eldar and S. Buzzi},
  title   = {Integrated sensing and communications: {T}oward dual-functional
             wireless networks for {6G} and beyond},
  journal = {IEEE J. Sel. Areas Commun.},
  volume  = {40},
  number  = {6},
  pages   = {1728--1767},
  year    = {2022},
  month   = jun,
  doi     = {10.1109/JSAC.2022.3156632}
}

@article{palhares2026,
  author        = {V. Palhares and A. Grudnitsky and S. Mandelli},
  title         = {Weather estimation for integrated sensing and communication},
  journal       = {arXiv preprint arXiv:2601.15145
        
        
        
        
        
        
        
        },
  year          = {2026},
  eprint        = {2601.15145},
  archiveprefix = {arXiv},
  primaryclass  = {eess.SP}
}

@article{weiss2024,
  author  = {T. Weiss and T. Routtenberg and J. Ostrometzky and H. Messer},
  title   = {Intensity estimation after detection for accumulated rainfall
             estimation},
  journal = {Frontiers Signal Process.},
  volume  = {4},
  pages   = {1291878},
  year    = {2024},
  doi     = {10.3389/frsip.2024.1291878}
}

@article{xiong2023,
  author  = {Y. Xiong and F. Liu and Y. Cui and W. Yuan
             and T. X. Han and G. Caire},
  title   = {On the fundamental tradeoff of integrated sensing and
             communications under {Gaussian} channels},
  journal = {IEEE Trans. Inf. Theory},
  volume  = {69},
  number  = {9},
  pages   = {5723--5751},
  year    = {2023},
  month   = sep,
  doi     = {10.1109/TIT.2023.3284449},
  note    = {2025 IEEE ComSoc-IT Joint Paper Award}
}

@article{dong2025debrisense,
  author  = {H. Dong and O. B. Akan},
  title   = {{DebriSense}: {THz}-Based Integrated Sensing and Communications ({ISAC})
             for Debris Detection and Classification in the {Internet of Space} ({IoS})},
  journal = {IEEE Trans. Wireless Commun.},
  volume  = {24},
  number  = {11},
  pages   = {9282--9295},
  year    = {2025},
  month   = nov,
  doi     = {10.1109/TWC.2025.3572276}
}

@inproceedings{sagiv2023,
  author    = {S. Sagiv and H. Messer},
  title     = {A {Cram\'{e}r--Rao} based study of {2-D} fields retrieval
               by measurements from a random sensor network},
  booktitle = {Proc. IEEE Int. Conf. Acoust., Speech, Signal Process.
               Workshops (ICASSPW)},
  year      = {2023},
  month     = jun,
  doi       = {10.1109/ICASSPW59220.2023.10193063}
}

@book{vantrees1968,
  author    = {H. L. {Van Trees}},
  title     = {Detection, Estimation, and Modulation Theory, {Part I}},
  publisher = {John Wiley \& Sons},
  address   = {New York, NY, USA},
  year      = {1968}
}

@article{gini1998,
  author  = {F. Gini},
  title   = {A radar application of a modified {Cram\'{e}r--Rao} bound:
             Parameter estimation in non-{Gaussian} clutter},
  journal = {IEEE Trans. Signal Process.},
  volume  = {46},
  number  = {7},
  pages   = {1945--1953},
  year    = {1998},
  month   = jul,
  doi     = {10.1109/78.700966
        
        
}}

@inproceedings{ghaddar2025active,
  title={Active Uplink Sensing Beamformer Design via Bayesian Cram{\'e}r-Rao Bound Dual Optimization},
  author={Ghaddar, Nadim and Yu, Wei},
  booktitle={ICC 2025-IEEE International Conference on Communications},
  pages={5736--5741},
  year={2025},
  organization={IEEE}
}

@article{kedem1987,
  author  = {B. Kedem and L. S. Chiu},
  title   = {On the lognormality of rain rate},
  journal = {Proc. Nat. Acad. Sci. USA},
  volume  = {84},
  number  = {4},
  pages   = {901--905},
  year    = {1987},
  month   = feb,
  doi     = {10.1073/pnas.84.4.901}
}

@article{chwala2019,
  author  = {Chwala, Christian and Kunstmann, Harald},
  title   = {Commercial microwave link networks for rainfall observation:
             {Assessment} of the current status and future challenges},
  journal = {WIREs Water},
  year    = {2019},
  volume  = {6},
  number  = {2},
  pages   = {e1337},
  doi     = {10.1002/wat2.1337}
}

@article{graf2020,
  author  = {Graf, Maximilian and Chwala, Christian and Polz, Julius
             and Kunstmann, Harald},
  title   = {Rainfall estimation from a {German}-wide commercial microwave
             link network: optimized processing and validation for 1 year
             of data},
  journal = {Hydrol. Earth Syst. Sci.},
  year    = {2020},
  volume  = {24},
  pages   = {2931--2950},
  doi     = {10.5194/hess-24-2931-2020}
}

@article{gill1995,
  author  = {Gill, Richard D. and Levit, Boris Y.},
  title   = {Applications of the van {Trees} inequality: A {Bayesian}
             {Cram\'{e}r--Rao} bound},
  journal = {Bernoulli},
  year    = {1995},
  volume  = {1},
  number  = {1--2},
  pages   = {59--79}
}

@article{tichavsky1998,
  author  = {Tichavsk\'{y}, Petr and Muravchik, Carlos H. and Nehorai, Arye},
  title   = {Posterior {Cram\'{e}r--Rao} bounds for discrete-time nonlinear
             filtering},
  journal = {IEEE Trans. Signal Process.},
  year    = {1998},
  volume  = {46},
  number  = {5},
  pages   = {1386--1396},
  month   = may,
  doi     = {10.1109/78.668800}
}

@article{aharon2024,
  author  = {Aharon, Ori and Tabrikian, Joseph},
  title   = {Asymptotically Tight {Bayesian} {Cram\'{e}r--Rao} Bound},
  journal = {IEEE Trans. Signal Process.},
  year    = {2024},
  volume  = {72},
  pages   = {3333--3346},
  doi     = {10.1109/TSP.2024.3421900}
}

@article{li2023leo,
  author  = {Li, Ke-Xin and Gao, Xiqi and Xia, Xiang-Gen},
  title   = {Channel Estimation for {LEO} Satellite Massive {MIMO} {OFDM}
             Communications},
  journal = {IEEE Trans. Wireless Commun.},
  year    = {2023},
  volume  = {22},
  number  = {11},
  pages   = {7537--7550},
  month   = nov,
  doi     = {10.1109/TWC.2023.3252895}
}

@article{moustakides1986,
  author  = {Moustakides, George V.},
  title   = {Optimal Stopping Times for Detecting Changes in Distributions},
  journal = {Ann. Statist.},
  year    = {1986},
  volume  = {14},
  number  = {4},
  pages   = {1379--1387},
  month   = dec,
  doi     = {10.1214/aos/1176350164}
}

@article{leyva2025,
  author={Leyva-Mayorga, Israel and Saggese, Fabio and Li, Lintao and Popovski, Petar},
  journal={IEEE Transactions on Wireless Communications}, 
  title={Integrating Atmospheric Sensing and Communications for Resource Allocation in NTNs}, 
  year={2025},
  volume={24},
  number={11},
  pages={9703-9718},
  doi={10.1109/TWC.2025.3574760}}

@ARTICLE{han2023cml5g,

  author={Han, Congzheng and Zhang, Gaoyuan and Ji, Baofeng and Huo, Juan and Messer, Hagit},

  journal={IEEE Communications Magazine}, 

  title={On the Potential of Using Emerging Microwave Links for City Rainfall Monitoring}, 

  year={2023},

  volume={61},

  number={11},

  pages={174-180},

  doi={10.1109/MCOM.001.2200975}}

@article{xu2024jsac,
  author  = {Xu, Chan and Zhang, Shuowen},
  title   = {{MIMO} Integrated Sensing and Communication Exploiting
             Prior Information},
  journal = {IEEE J. Sel. Areas Commun.},
  year    = {2024},
  volume  = {42},
  number  = {9},
  pages   = {2306--2321},
  month   = sep,
  doi     = {10.1109/JSAC.2024.3413972}
}

@article{hassibi2003,
  author  = {B. Hassibi and B. M. Hochwald},
  title   = {How much training is needed in multiple-antenna wireless links?},
  journal = {IEEE Trans. Inf. Theory},
  volume  = {49},
  number  = {4},
  pages   = {951--963},
  year    = {2003},
  month   = apr,
  doi     = {10.1109/TIT.2003.809594}
}

@article{humphreys2023,
  author  = {T. E. Humphreys and P. A. Iannucci
             and Z. M. Komodromos and A. M. Graff},
  title   = {Signal structure of the {Starlink} {Ku}-band downlink},
  journal = {IEEE Trans. Aerosp. Electron. Syst.},
  volume  = {59},
  number  = {5},
  pages   = {6016--6030},
  year    = {2023},
  month   = oct,
  doi     = {10.1109/TAES.2023.3268610
        
        
}
}

@article{page1954,
  author  = {E. S. Page},
  title   = {Continuous inspection schemes},
  journal = {Biometrika},
  volume  = {41},
  number  = {1--2},
  pages   = {100--115},
  year    = {1954},
  month   = jun,
  doi     = {10.1093/biomet/41.1-2.100
        
        
}}

@article{lorden1971,
  author  = {G. Lorden},
  title   = {Procedures for reacting to a change in distribution},
  journal = {Ann. Math. Statist.},
  volume  = {42},
  number  = {6},
  pages   = {1897--1908},
  year    = {1971},
  month   = dec,
  doi     = {10.1214/aoms/1177693055
        
        
}}

@unpublished{dong2026eumc,
  author = {H. Dong and O. B. Akan},
  title  = {{Ku}-band satellite signals for rain rate estimation:
            Performance limits and 215-link validation},
  note   = {Submitted to Eur. Microw. Conf. (EuMC) 2026, London, UK},
  year   = {2026}
}

@article{wang2025e2l,
  author        = {H. Wang and H. Dong and H. Cai and O. B. Akan},
  title         = {Environment-to-Link {ISAC} with Space-Weather Sensing for {Ka}-Band {LEO} Downlinks},
  journal       = {arXiv preprint arXiv:2601.00820},
  year          = {2026},
  eprint        = {2601.00820},
  archiveprefix = {arXiv},
  primaryclass  = {eess.SP}
}



\end{document}